\documentclass[11pt, draftclsnofoot, journal, letterpaper, onecolumn]{IEEEtran}

\usepackage[dvips]{graphicx}
\usepackage{times}
\usepackage{cite}
\usepackage{amsmath}
\usepackage{array}
\usepackage{amssymb}

\usepackage{bbm}

\usepackage{stfloats}
\usepackage{rotating,threeparttable,booktabs}
\usepackage{bm}
\usepackage{dcolumn,booktabs}
\usepackage{multirow}
\usepackage{graphicx}
\usepackage{subfigure}
\usepackage{color}

\usepackage[letterpaper]{geometry}

\usepackage{amsmath,amsthm,amssymb,amsfonts}
\makeatletter
\thm@headfont{\sc}
\makeatother
\newtheorem{theorem}{Theorem}

\newtheorem{lemma}{Lemma}

\begin{document}

\title{\LARGE {Channel Estimation and Optimal Training Design for Correlated MIMO Two-Way Relay Systems in Colored Environment}}
\author{\IEEEauthorblockN{Rui~Wang, Meixia Tao, \textit{Senior Member, IEEE}, Hani Mehrpouyan, \textit{Member, IEEE}, and Yingbo Hua, \textit{Fellow, IEEE}\vspace{-20pt}}
\thanks{R. Wang is with the Department of Information and Communications at Tongji University, Shanghai, 201804, P. R. China. Email: liouxingrui@gmail.com. M. Tao is with the Department of Electronic Engineering at Shanghai Jiao Tong University, Shanghai, China. Emails: mxtao@sjtu.edu.cn. H. Mehrpouyan is with the Department of Computer and Electrical Engineering and Computer Science at California State University, Bakersfield, CA, USA. Email: hani.mehr@ieee.org. Y. Hua is with the Department of Electrical Engineering at the University of California, Riverside, CA, USA. Email: yhua@ee.ucr.edu.}}

\maketitle
\vspace{-1cm}
%-------------------------------------------------------------Abstract-----------------------------------------------------------------------
\begin{abstract}

In this paper, while considering the impact of antenna correlation and the interference from neighboring users, we analyze channel estimation and training sequence design for multi-input multi-output (MIMO) two-way relay (TWR) systems.  To this end, we propose to decompose the bidirectional transmission links into two phases, i.e., the multiple access (MAC) phase and the broadcasting (BC) phase. By considering the Kronecker-structured channel model, we derive the optimal linear minimum mean-square-error (LMMSE) channel estimators. The corresponding training designs for the MAC and BC phases are then formulated and solved to improve channel estimation accuracy. For the general scenario of training sequence design for both phases, two iterative training design algorithms are proposed that are verified to produce training sequences that result in near optimal channel estimation performance. Furthermore, for specific practical scenarios, where the covariance matrices of the channel or disturbances are of particular structures, the optimal training sequence design guidelines are derived. In order to reduce training overhead, the minimum required training length for channel estimation in both the MAC and BC phases are also derived. Comprehensive simulations are carried out to demonstrate the effectiveness of the proposed training designs.

%For some special cases where the transmit channel covariance or the temporal covariance matrix of disturbance has specific forms, we derive optimal structures of the training sequences.
%%Besides that the training sequence design,
%The required minimum training lengths for both phases are also analyzed to support the optimal training designs.
%Comprehensive simulations
%are provided to demonstrate the effectiveness of the proposed
%training designs.
\end{abstract}

%\begin{IEEEkeywords}
%MIMO, two-way relaying, channel estimation, linear minimum mean-square-error, convex optimization.
%\end{IEEEkeywords}

\section{Introduction}
Relay assisted cooperative communications has been regarded as one of the most promising techniques in combating long distance channel fading in complex wireless communication systems. One popular example is one-way relaying, which has been well studied in the past decade \cite{Hua_JSAC_12, Rong2009,Ting2011}. Although one-way relaying shows great potential in reducing power consumption, enhancing reliability, and extending coverage, it suffers from low spectral efficiency due to the half-duplex nature of the network. To overcome this disadvantage, by using the idea of network coding, two-way relaying (TWR) has been proposed and has received great attention recently \cite{ZhangSL_PLNC06}. In fact, TWR can maintain the advantages of traditional relaying while doubling spectrum efficiency.

The improvement in spectrum efficiency in TWR is achieved by applying self-interference cancelation at each source node and extracting the desired information from the received network-coded messages. In this case, the accuracy of the self-interference cancelation process significantly affects the performance of TWR systems. Moreover, when using the popular amplify-and-forward (AF) relaying strategy, the accuracy of self-interference cancelation process is highly dependent on the precision of the channel estimation process. Thus, obtaining highly accurate channel state information (CSI) becomes more important in TWR systems compared to traditional one-way relaying systems. In fact, devising new channel estimation schemes for TWR systems has received great attention recently. For example, in \cite{Gao_tcom2009}, the authors propose to estimate the cascaded channel of TWR systems under the AF relaying strategy. By using multiple phase shift keying (M-PSK) training symbols, blind and partially-blind channel estimators are investigated in \cite{Abdallah2012twc, Abdallah2012tsp}. Different from \cite{Gao_tcom2009, Abdallah2012twc, Abdallah2012tsp}, where flat fading channel are assumed, the authors in \cite{Gongpu_twc2011} investigate time varying channel estimation via a new complex-exponential basis expansion model. Moreover, in \cite{Gao_tsp2009, Shun_tsp2012}, the channel estimation process for TWR is extended to the scenario of orthogonal frequency division multiplexing (OFDM) systems.

It is worth noting that the works summarized above are concerned with single-antenna TWR systems. As expected, the multi-antenna or multi-input multi-output (MIMO) technique can be introduced to TWR systems to further improve transmission reliability and bandwidth efficiency. One efficient way to realize such performance improvement is to exploit the estimated CSI for the application of source and relay precoding
\cite{Zhang2009, Rong_tsp2012, RuiWang, Xu_Hua_TWC}.
%\cite{Zhang2009, Rong_tsp2012, RuiWang, RuiTVT2012, WangTWC2012, Tao2012}.
Therefore, in MIMO TWR systems, in addition to affecting the performance of self-interference cancelation, inaccurate channel estimation also imposes a negative effect on the precoder design.

Fig. \ref{SystemModel} depicts a MIMO TWR setup. Let us denote the process of data transmission from the source nodes to the relay and relay to the source nodes as the broadcasting (BC) and multiple access (MAC) phases, respectively. In \cite{Jian_GC2008}, a MIMO channel estimator is proposed that uses the self-interference as a training sequence to estimate the channel matrices corresponding to the BC phase. In \cite{Zhaoxi_2011}, the performances of different channel estimators, including individual and cascaded channel estimators, are compared based on the least squares (LS) criterion. In \cite{Pham_tcom2010}, an LS estimator is used to obtain the cascaded channel matrices corresponding to the BC and MAC phases using a single carrier cyclic prefix. Note that in the contributions of \cite{Jian_GC2008, Zhaoxi_2011, Pham_tcom2010}, the channel statistics, whether cascaded channels or the individual channels, are assumed to be unknown deterministic matrices. Based on the estimation theory, if channel statistics are known, the channel estimation can be conducted under the Bayesian framework and the estimation accuracy can be further enhanced. Hence, by taking these statistics into account, we seek to improve upon the channel estimators in \cite{Jian_GC2008, Zhaoxi_2011, Pham_tcom2010}.

Very recently, the authors in \cite{Kim_cl2013, Rong_tsp2013} independently investigate the minimum mean-square-error (MMSE) channel estimation for TWR systems based on a correlated Gaussian MIMO channel model. In particular, in \cite{Kim_cl2013}, the cascaded channel matrices for AF TWR systems are estimated and the training sequences at the two source nodes are optimized to minimize the total channel estimation MSE. Different from \cite{Kim_cl2013}, the authors in \cite{Rong_tsp2013} aim to estimate the individual channel matrices for each link. To reach this goal, two different estimation schemes, i.e., the superimposed channel training and the two-stage channel estimation schemes, are proposed. In addition, the training sequences at the two source nodes, as well as, at the relay node are jointly optimized to improve channel estimation accuracy.

\begin{figure}[t]
\begin{centering}
\includegraphics[scale=0.65]{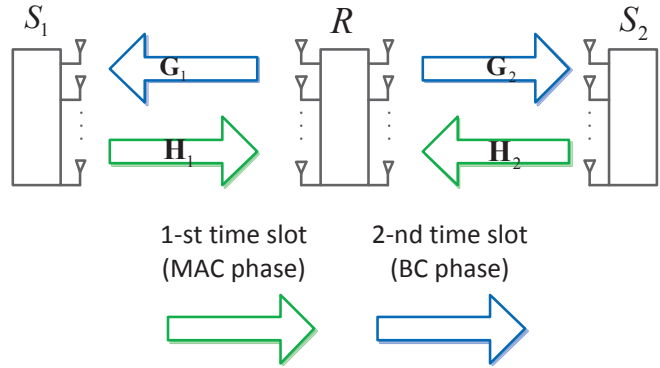}
\vspace{-0.1cm}
\caption{An illustration of MIMO two-way relay system.} \label{SystemModel}
\end{centering}
\vspace{-0.5cm}
\end{figure}
In this paper, similar to \cite{Jian_GC2008, Zhaoxi_2011, Pham_tcom2010, Kim_cl2013, Rong_tsp2013}, while assuming that the channel statistics are known, we analyze and devise channel estimators for correlated MIMO TWR systems. Specifically, we consider the Kronecker-structured channel model, such that the individual channel matrices can be estimated based on the Bayesian framework. However, \textit{unlike} \cite{Kim_cl2013, Rong_tsp2013}, we take into account the interference from the nearby users. Thus, in this model, the disturbance at the source nodes and the relay node consists of both noise and interference. Note that the considered colored estimation environment may be more practical for applications in today's more densely deployed wireless networks. Although channel estimation in point-to-point MIMO systems in colored environments has been studied in \cite{Biguesh_tsp2009, Bjornson2010}, to the best of our knowledge, this topic has not been addressed in the TWR scenario.

To enhance TWR performance, we seek to estimate the individual channel matrices corresponding to source-to-relay and relay-to-source links, see Fig. \ref{SystemModel}. To this end, a new two-phase estimation scheme is proposed, where the bidirectional transmission of a TWR system is decomposed into the MAC and BC phases. For the MAC and BC phases, the channel estimation is performed at the relay node and two source/user nodes, respectively. The proposed estimation scheme is different from the ones in \cite{Kim_cl2013, Rong_tsp2013}, where the channel estimation is assumed to only be conducted at the user ends. As such, our proposed estimation scheme can more efficiently support precoding at the relay since it requires significantly less feedback overhead \cite{RuiWang, RuiTVT2012, WangTWC2012}. Based on the proposed estimation scheme, we derive the optimal linear MMSE (LMMSE) estimator for each phase. Next, the corresponding training design problems are formulated with the aim of minimizing the total MSE of channel estimation process for each phase. The training design problem considered here is different from that of \cite{Kim_cl2013, Rong_tsp2013}, since we take into account the effect of colored disturbances caused by interference at the relay node and user ends. Moreover, the training design scenarios for point-to-point systems in \cite{Biguesh_tsp2009, Bjornson2010} are different from the scenario under consideration in this paper, since our proposed training sequence design is optimized to simultaneously enhance channel estimation accuracy over both links in the BC and MAC phases. Although, for the general scenario, it is difficult to derive the optimal training sequence structures as in \cite{Kim_cl2013, Rong_tsp2013, Biguesh_tsp2009, Bjornson2010}, we propose two iterative design algorithms to solve the training design problems. These algorithms are verified to converge quickly to the near optimal solution and to not be sensitive to the initialization process. For some special cases, where the covariance matrices of the channels or disturbances have specific forms, two specific approaches are applied to obtain the optimal training sequences: 1) the original problem is converted into a standard convex optimization problem; 2) the optimal structures of the training sequences are first derived and then used to reduce the original non-convex problem into a simple power allocation problem. Finally, to reduce training overhead, the minimum required training length for channel estimation in both the MAC and BC phases are derived and extensive simulations are carried out to support the findings of the paper.

%The optimal lengths of the training sequences are also analyzed to support the proposed training designs.
%Simulation results verify the effectiveness of the proposed designs.
%  \item The optimality of the propose designs is verified by the extensive simulation results.
%\end{itemize}

The rest of the paper is organized as follows. In Section II, we present the system model.
The optimal LMMSE estimators for both MAC and BC phases are derived in Section III. The training designs for the MAC and the BC phases are analyzed in Section IV and V, respectively. Simulation results are provided in Section VI.
Finally, we conclude the paper in Section VII.

\emph{Notations}: $\cal E(\cdot)$ denotes the expectation operator.
$\otimes$ denotes the Kronecker operator. ${\rm vec}(\cdot)$ signifies the matrix vectorization operator.
Superscripts ${\bf A}^T$, ${\bf A}^{*}$, and ${\bf A}^H$ denote the transpose, conjugate, and conjugate transpose of matrix $\bf A$, respectively.
${\rm Tr}({\bf A})$, ${\bf A}^{-1}$, $\det({\bf A})$, and ${\rm Rank}(\bf A)$ stand for the trace, inverse, determinant, and rank of ${\bf A}$, respectively.
${\bm \lambda}({\bf A})$ denotes a vector containing eigenvalues of ${\bf A}$.
${\rm Blkdiag}({\bf A}_1,{\bf A}_2,\cdots,{\bf A}_N)$ denotes a block diagonal matrix constructed by matrices ${\bf A}_i$, for $\forall i$.
${\rm Diag}(\bf a)$ denotes a diagonal matrix with ${\bf a}$ being its diagonal entries.
${\bf A}(n:m,:)$ and ${\bf A}(:,n:m)$ denote the sub-matrices constructed by $n$ to $m$ rows and $n$ to $m$ columns of ${\bf A}$, respectively.
$||{\bf A}||^2_F$ denotes the Frobenius norm of ${\bf A}$.
%
%${\bf 0}_{N\times M}$ implies the $N\times M$ zero matrix and ${\bf I}_N$ denotes the $N \times N$ identity matrix.
${\bf 0}$ and ${\bf I}$ denote the zero and identity matrices, respectively.
%$|z|$ implies the norm of the complex number $z$, ${\Re}(z)$ and $\Im (z)$ denote the real and imaginary part of $z$, respectively.
${\Re}(z)$ denotes the real part of complex variable $z$.
%$||{\bf x}||^2_2$ denotes the squared Euclidean norm of a complex vector ${\bf x}$ and
%
The distribution of a circular symmetric complex Gaussian vector with mean vector $\bf x$ and covariance matrix ${\bf \Sigma} $ is denoted by ${\cal CN}({\bf x},{\bf \Sigma})$.
${\mathbb C}^{x \times y}$ denotes the space of complex $x \times y$ matrices. $\mathbb{S}^N$ and $\mathbb{S}^N_{+}$ denote the set of symmetric $N \times N$ matrices and the set of positive semidefinite $N \times N$ matrices, respectively. ${\bf x} \preccurlyeq {\bf y}$ denotes that the vector $\bf y$ majorizes the vector $\bf x$.

\section{System Model}
Consider a TWR system, where source nodes $S_1$ and $S_2$ intend to exchange messages with one another through a relay node $R$. $S_1$, $R$, and $S_2$ are assumed to be equipped with $N_1$, $M$, and $N_2$ antennas, respectively. The channel matrices from $S_1$ and $S_2$ to $R$ are denoted by ${\bf H}_1$ and ${\bf H}_2$, respectively, and the channel matrices from $R$ to $S_1$ and $S_2$ are denoted by ${\bf G}_1$ and ${\bf G}_2$, respectively.

Signal transmission within the TWR system is assumed to be achieved in two time slots. In the first phase, referred to as the MAC phase, the source node $S_i$, for $i=1,2$, transmits its signal to the relay node $R$, while in the second phase, referred to as the BC phase, the relay node $R$ forwards the its combined received signal to the two source nodes $S_1$ and $S_2$. The proposed channel estimator aims to obtain the individual channels of the two hops, i.e., $\{{\bf H}_1,{\bf H}_2,{\bf G}_1,{\bf G}_2\}$. Note that different from the cases studied in \cite{Gao_tcom2009,Kim_cl2013}, where the cascaded channels are estimated, here, the individual channel matrices are estimated. This approach enhances precoding design and/or power allocation at the relay node, which can further improve the overall system performance \cite{Zhang2009, Rong_tsp2012, RuiWang, Xu_Hua_TWC,RuiTVT2012, WangTWC2012}.

Following the transmission model in Fig. \ref{SystemModel}, we assume that ${\bf H}_1$ and ${\bf H}_2$ are estimated in the MAC phase via the training signals sent from the two source nodes, and ${\bf G}_1$ and ${\bf G}_2$ are estimated in the BC phase via the training signal transmitted from the relay.

The received training signals at the relay node in the MAC phase can be expressed as
\begin{equation}\label{EQU-1}
{\bf y}_R (t)= {\bf H}_1 {\bf s}_1 (t)+ {\bf H}_2 {\bf s}_2(t)  + {\bf n}_R(t),
\end{equation}
where ${\bf s}_i (t)\in \mathbb{C}^{N_i \times 1}$ denotes the training signal at the source $S_i$ and ${\bf n}_R(t) \in \mathbb{C}^{M \times 1}$ represents the correlated Gaussian disturbance at the relay node. ${\bf n}_R(t)$ models the total background noise as well as the interference from adjacent communication links. ${\bf n}_R(t)$ is modeled as a stochastic process with respect to the time variable $t$ \cite{Biguesh_tsp2009,Bjornson2010,Yong_tsp2007}. Here, the channel matrix ${\bf H}_i \in \mathbb{C}^{M \times N_i}$ is modeled by the Rayleigh fading with mean zero and covariance ${\bf Z}_{H_i} \in \mathbb{S}^{MN_i \times MN_i}_+$, i.e., ${\rm vec}({\bf H}_i )\sim \mathcal{CN}({\bf 0}, {\bf Z}_{H_i})$. To estimate the channel matrices at the relay, the source nodes typically need to send a sequence of known training signals. Assuming training sequences have a length of $L_S$, the received signal in \eqref{EQU-1} can be written in matrix from as\vspace{-0pt}
\begin{equation}\label{EQU-2}
{\bf Y}_R = {\bf H}_1 {\bf S}_1 + {\bf H}_2 {\bf S}_2  + {\bf N}_R,
\end{equation}
where ${\bf Y}_R \triangleq [{\bf y}_R (1), {\bf y}_R (2), \cdots, {\bf y}_R (L_S)] \in \mathbb{C}^{M \times L_S}$, ${\bf S}_i \triangleq [{\bf s}_1 (1), {\bf s}_1 (2), \cdots, {\bf s}_1 (L_S)] \in \mathbb{C}^{N_i \times L_S}$ and ${\bf N}_R \triangleq [{\bf n}_R(1), {\bf n}_R(2), \cdots, {\bf n}_R(L_S)] \in \mathbb{C}^{M \times L}$. Here, the disturbance ${\bf N}_R$ is modeled as ${\rm vec}({\bf N}_R)\sim \mathcal{CN}({\bf 0}, {\bf K}_R)$ with ${\bf K}_R \in \mathbb{S}^{ML_S \times ML_S}_+$.
Suppose that the source node $S_i$ has a maximum power of $\tau_i$ during the channel estimation phase, the training sequence ${\bf S}_i$ should fulfill the following power constraint
\begin{equation}\label{EQU-3}
{\rm Tr} ({\bf S}_i {\bf S}^H_i) \leq \tau_i.
\end{equation}

In the BC phase, the received training signals at the source nodes are given by
\begin{equation}\label{EQU-4}
{\bf y}_i (t)= {\bf G}_i {\bf s}_R(t) + {\bf n}_i(t),~i=1,2
\end{equation}
where ${\bf s}_R (t)\in \mathbb{C}^{M \times 1}$ denotes the training signal at the relay node and ${\bf n}_i(t) \in \mathbb{C}^{N_i \times 1}$ represents the correlated Gaussian disturbance at the node $S_i$. The channel matrix ${\bf G}_i \in \mathbb{C}^{M \times N_i}$ is modeled by a Rayleigh fading parameter with mean zero and covariance ${\bf Z}_{G_i} \in \mathbb{S}^{MN_i \times MN_i}_+$, i.e., ${\rm vec}({\bf G}_i )\sim \mathcal{CN}({\bf 0}, {\bf Z}_{G_i})$. As in \eqref{EQU-1}, here, the disturbance term ${\bf n}_i(t)$ also includes the total background noise and interference from nearby communication nodes. By rewriting \eqref{EQU-4} into matrix form, we have
\begin{equation}\label{EQU-5}
\begin{split}
{\bf Y}_i = {\bf G}_i {\bf S}_R + {\bf N}_i,~i=1,2,
\end{split}
\end{equation}
where ${\bf Y}_i \triangleq [{\bf y}_i (1), {\bf y}_i (2), \cdots, {\bf y}_i (L_R)] \in \mathbb{C}^{N_i \times L_R}$, ${\bf S}_R \triangleq [{\bf s}_R (1), {\bf s}_R (2), \cdots, {\bf s}_R (L_R)] \in \mathbb{C}^{N_i \times L_R}$ and ${\bf N}_i \triangleq [{\bf n}_i(1), {\bf n}_i(2), \cdots, {\bf n}_i(L_R) \in \mathbb{C}^{M \times L_R}$. The disturbance ${\bf N}_i$ is modeled as ${\rm vec}({\bf N}_i)\sim \mathcal{CN}({\bf 0}, {\bf K}_i)$ with ${\bf K}_i \in \mathbb{S}^{N_i L_R \times N_i L_R}$. Here, we assume that the training sequence length at the relay is $L_R$. The following condition must be met to satisfy the power constraint at the relay node\vspace{-12pt}
\begin{equation}\label{EQU-55}
\begin{split}
{\rm Tr}({\bf S}_R {\bf S}^H_R) \leq \tau_R.
\end{split}
\end{equation}
In \eqref{EQU-55}, $\tau_R$ denotes the maximum power at the relay node during the training phase.

In this work, we assume that the channel covariance matrices ${\bf Z}_{H_i}$ and ${\bf Z}_{G_i}$ and the covariance of disturbance ${\bf K}_R$ and ${\bf K}_i$, for $i=1,2$, are structured and their statistics are known. Let us first focus on the properties of the channel statistics.

The correlation amongst the channel parameters can be caused by insufficient antenna spacing as verified by measurements in \cite{Chizhik_JSAC2003,Kai_VTC2002}. Accordingly, the channel matrices are assumed to follow the \emph{Kronecker-structured} model, i.e., the covariance matrices are separated between the transmitter and receiver sides and given by ${\bf Z}_{H_i} = {\bf Z}_{t,{H_i}} \otimes {\bf Z}_{r,H} $,  ${\bf Z}_{G_i} = {\bf Z}_{t,G} \otimes {\bf Z}_{r,{G_i}} $ for $j={1,2}$. Here, indexes `$t$' and `$r$' denote `transmitter' and `receiver', respectively. In addition, since the channels ${\bf H}_1$ and ${\bf H}_2$ terminate at the relay node and the channels ${\bf G}_1$ and ${\bf G}_2$ begin at the relay node, we have that ${\bf Z}_{r,H_1}={\bf Z}_{r,H_2} = {\bf Z}_{r,H}$ and ${\bf Z}_{t,G_1} = {\bf Z}_{t,G_2} = {\bf Z}_{t,G}$. Using the Kronecker model and the above definitions, the channel matrices can be expressed as \cite{Kim_cl2013,Rong_tsp2013,Biguesh_tsp2009,Bjornson2010,Yong_tsp2007}
\begin{equation}\label{EQU-6}
\begin{split}
{\bf H}_i = {\bf C}_{r,H} {\bf W}_{H_i} {\bf C}^T_{t,H_i},~
{\bf G}_i = {\bf C}_{r,G_i} {\bf W}_{G_i} {\bf C}^T_{t,G},~i=1,2
\end{split}
\end{equation}
where ${\bf Z}_{t,H_i} = {\bf C}_{t,H_i}  {\bf C}^H_{t,H_i} $, ${\bf Z}_{r,H} = {\bf C}_{r,H}  {\bf C}^H_{r,H} $, ${\bf Z}_{t,G} = {\bf C}_{t,G}  {\bf C}^H_{t,G} $, and ${\bf Z}_{r,G_i} = {\bf C}_{r,G_i}  {\bf C}^H_{r,G_i} $. ${\bf W}_{H_i}$ and ${\bf W}_{G_i}$ are unknown matrices, where their entries are modeled by $\mathcal{CN}(0, 1)$.

The structured disturbance covariance ${\bf K}_i$, for $i\in \{R, 1,2 \}$, are assumed to be modeled by \cite{Biguesh_tsp2009,Bjornson2010,Yong_tsp2007}\vspace{-6pt}
\begin{equation}\label{EQU-7}
\begin{split}
{\bf K}_i = {\bf K}_{q,i} \otimes {\bf K}_{r,i},~i=R, 1,2,
\end{split}
\end{equation}
where ${\bf K}_{q,1} / {\bf K}_{q,2} \in \mathbb{C}^{L_S \times L_S}$, ${\bf K}_{q,R} \in \mathbb{C}^{L_R \times L_R}$ denote the temporal covariance matrix and ${\bf K}_{r,1} \in \mathbb{C}^{N_1 \times N_1}$, ${\bf K}_{r,2} \in \mathbb{C}^{N_2 \times N_2}$, and ${\bf K}_{r,R} \in \mathbb{C}^{M \times M}$ denote the received spatial covariance
matrix. Moreover, it is assumed that ${\bf K}_{r,1}$, ${\bf K}_{r,2}$, and ${\bf K}_{r,R}$ share the same eigenvectors with ${\bf Z}_{r,{G_1}}$, ${\bf Z}_{r,{G_2}}$ and ${\bf Z}_{r,H}$, respectively. This assumption is valid when the disturbances are either spatially uncorrelated or share the same spatial structure as the channel \cite{Yong_tsp2007, Bjornson2010}. In addition, as summarized in \cite{Bjornson2010}, this assumption models the following scenarios: 1) Additive noise-limited scenario, ${\bf K}_{r,i} = \mu_i {\bf I}$ with $\mu_i$ being some variance, for $i=R, 1,2$; 2) Interference-limited scenario, ${\bf K}_{r,i} = {\bf Z}_{r,{G_i}}$, for $i=1,2$, and ${\bf K}_{r,R}={\bf Z}_{r,H}$;
3) Additive noise and temporally uncorrelated interference scenario, ${\bf K}_{r,i} = \mu_i {\bf I}+\nu_i {\bf Z}_{r,{G_i}}$, for $i=1,2$, and ${\bf K}_{r,R} = \mu_R {\bf I}+\nu_R{\bf Z}_{r,H}$ with $\nu_i$, for $i=R,1,2$, being the number of interfering users; and 4) Additive noise and spatially uncorrelated interference scenario, ${\bf K}_{r,i} = {\bf I}$.

The singular value decomposition (SVD) of ${\bf Z}_{t,{H_i}}$, ${\bf Z}_{r,H}$, ${\bf Z}_{t,G}$, and ${\bf Z}_{r,{G_i}}$ are given by
\begin{equation}\label{EQU-8}
\begin{split}
{\bf Z}_{a,b} =  {\bf U}_{a,b} {\bm \Sigma}_{a,b} {\bf U}^H_{a,b},~
 a \in \{r,t\},b \in \{H,H_1,H_2, G, G_1, G_2\},
\end{split}
\end{equation}
where ${\bf U}_{a,b}$ denotes the unitary eigenvector matrix and ${\bm \Sigma}_{a,b}$ is a diagonal matrix with $[{\bm \Sigma}_{a,b}]_{n,n} ={\sigma}_{a,b,n}$ being the $n$-th eigenvalue of ${\bf Z}_{a,b}$.
Accordingly, the SVD decomposition of ${\bf C}_{a,b}$ is denoted by ${\bf C}_{a,b} =  {\bf U}_{a,b} {\bm \Sigma}^{1/2}_{a,b} \tilde{{\bf U}}^H_{a,b}$ with $\tilde{{\bf U}}_{a,b}$ represting a unitrary matrix. The SVD decomposition of ${\bf K}_{q,i}$ and ${\bf K}_{r,i}$ is denoted by \vspace{-0pt}
\begin{equation}\label{EQU-8}
%\begin{split}
{\bf K}_{a,b} =  {\bf V}_{a,b} {\bm \Delta}_{a,b} {\bf V}^H_{a,b},~
 a \in \{r,t\},b \in \{1,2, R\},
%\end{split}
\end{equation}
where ${\bf V}_{a,b}$ denotes the unitary eigenvector matrix, ${\bm \Delta}_{a,b}$ is a diagonal matrix with $[{\bm \Delta}_{a,b}]_{n,n} ={\delta}_{a,b,n}$ being the $n$-th eigenvalue of ${\bf K}_{a,b}$. As mentioned before, it is assumed that ${\bf V}_{r,1} = {\bf U}_{r,G_1}$, ${\bf V}_{r,2} = {\bf U}_{r,G_2}$ and ${\bf V}_{r,R} = {\bf U}_{r,H}$.

\vspace{-0pt}
\section{Channel Estimation for Two-Way Relay Systems}

Following the proposed estimation scheme in Section II, we next obtain the channel estimates based on \eqref{EQU-2} and \eqref{EQU-5}.
For the estimation during the MAC phase, we rewrite \eqref{EQU-2} as
\begin{align}\label{III-1}
%\begin{split}
{\bf Y}_R = {\bf C}_{r,H} {\bf W}_{H_1} {\bf C}^T_{t,H_1} {\bf S}_1 + {\bf C}_{r,H} {\bf W}_{H_2} {\bf C}^T_{t,H_2} {\bf S}_2  + {\bf N}_R
 = {\bf C}_{r,H} {\bf W}_{H} {\bf C}^T_{t,H} {\bf S} + {\bf N}_R,
%\end{split}
\end{align}
where
${\bf W}_{H} \triangleq [{\bf W}_{H_1}, {\bf W}_{H_2}]$, ${\bf C}^T_{t,H} \triangleq {\rm Blkdiag}({\bf C}^T_{t,H_1}, {\bf C}^T_{t,H_2})$, and ${\bf S} \triangleq [{\bf S}^T_1, {\bf S}^T_2]^T$. Vectorizing ${\bf Y}_R$ in \eqref{III-1} and applying the identity\vspace{-12pt}
\begin{align}\label{III-1-1}
\begin{split}
{\rm vec}({\bf A}{\bf B}{\bf C})=({\bf C}^T\otimes {\bf A}){\rm vec}({\bf B}),
\end{split}
\end{align}
\vspace{-30pt}\\
we can rewrite \eqref{III-1} into\vspace{-12pt}
\begin{equation}\label{III-2}
\begin{split}
{\bf y}_R = \left({\bf S}^T {\bf C}_{t,H}\otimes {\bf C}_{r,H} \right) {\bf w}_{H} + {\bf n}_R,
\end{split}
\end{equation}
\vspace{-24pt}\\
where ${\bf y}_R \triangleq{\rm vec}({\bf Y}_R)$, ${\bf w}_{H} \triangleq {\rm vec}({\bf W}_{H})$ and ${\bf n}_R \triangleq {\rm vec}({\bf N}_R)$. The estimation of ${\bf w}_{H}$ based on the LMMSE criterion can be obtained as $\hat{{\bf w}}_{H} = {\bf T}_R {\bf y}_R$. The estimation matrix ${\bf T}_R$ has the following form \cite{StevenBook}
\begin{equation}\label{III-4}\vspace{-12pt}
\begin{split}
{\bf T}_R = \mathfrak{R}_{w_h y_R} \mathfrak{R}^{-1}_{y_R y_R}.
\end{split}
\end{equation}
where
\begin{equation}\label{III-5}\nonumber
\begin{split}
\mathfrak{R}_{w_h y_R} \triangleq& {\cal E} ({\bf w}_{H} {\bf y}^H_R)
 =   {\bf C}^H_{t,H} {\bf S}^* \otimes {\bf C}^H_{r,H}\\
\mathfrak{R}_{y_R y_R}  \triangleq& {\cal E} ({\bf y}_R {\bf y}^H_R)
 = \left({\bf S}^T {\bf C}_{t,H}\otimes {\bf C}_{r,H} \right) \left({\bf S}^T {\bf C}_{t,H}\otimes {\bf C}_{r,H} \right)^H + {\bf K}_R
 = {\bf S}^T {\bf C}_{t,H} {\bf C}^H_{t,H} {\bf S}^* \otimes {\bf C}_{r,H} {\bf C}^H_{r,H} + {\bf K}_R.
\end{split}
\end{equation}
Let us define ${\bf h} \triangleq {\rm vec}({\bf H}) = \left( {\bf C}_{t,H} \otimes {\bf C}_{r,H} \right) {\bf w}$
with ${\bf H} \triangleq [{\bf H}_1,{\bf H}_2]$,
the resulting estimation error, or mean-square-error (MSE), $e_R$ can be derived as
\begin{equation}\label{III-6}
\begin{split}
e_R  =& {\cal E}\left(||{\bf h} - \hat{{\bf h}} ||^2_2 \right)\notag\\
=& {\rm Tr}\left[ {\bf C}_{0,H} ({\bf w}_{H} - \hat{{\bf w}}_{H})({\bf w}_{H} - \hat{{\bf w}}_{H})^H \right]\notag\\
 =& {\rm Tr}\left[ {\bf C}_{0,H} ({\bf w}_{H} - {\bf T}_R {\bf y}_R)({\bf w}_{H} - {\bf T}_R {\bf y}_R)^H \right],
\end{split}
\end{equation}
where ${\bf C}_{0,H} \triangleq {\bf C}^H_{t,H} {\bf C}_{t,H} \otimes {\bf C}^H_{r,H} {\bf C}_{r,H} $.
Substituting ${\bf T}_R$ into \eqref{III-6} and using the matrix identity $({\bf I} + {\bf A}{\bf B})^{-1} = {\bf I}- {\bf A}({\bf I}+{\bf B}{\bf A})^{-1}{\bf B}$, we obtain the following more compact form for $e_R$
\begin{equation}\label{III-7}
%\begin{split}
e_R   = {\rm Tr}\left[ {\bf C}_{0,H} \left({\bf I} + \left({\bf S}^T {\bf C}_{t,H}\otimes {\bf C}_{r,H} \right)^H {\bf K}^{-1}_R
 \left({\bf S}^T {\bf C}_{t,H}\otimes {\bf C}_{r,H} \right)  \right )^{-1} \right].
%\end{split}
\end{equation}
%(\textcolor{red}{!Here the derivation is some kind of complicated. However, this result is well-know. The similar result can be found in \cite{StevenBook} on page 391 eq. (12.29).})
Note that since the channel estimation model in \eqref{III-2} is linear and Gaussian, the proposed LMMSE estimator is equivalent to the optimal MMSE estimator.

During the BC phase, the channel estimation should be based on the received signal in \eqref{EQU-5}. By vectorizing ${\bf Y}_i$ in \eqref{EQU-5}, we have
\begin{equation}\label{III-71}
\begin{split}
{\bf y}_i = \left({\bf S}^T_R {\bf C}_{t,G}\otimes {\bf C}_{r,G_i} \right) {\bf w}_{G_i} + {\bf n}_i,~i=1,2
\end{split}
\end{equation}
where ${\bf y}_i \triangleq{\rm vec}({\bf Y}_i)$, ${\bf w}_{G_i} \triangleq {\rm vec}({\bf W}_{G_i})$, and ${\bf n}_i \triangleq {\rm vec}({\bf N}_i)$. Similar to steps above, the estimation MSE of ${\bf g}_i$ can be obtained as
\begin{equation}\label{III-72}
\begin{split}
e_i & = {\cal E}\left(||{\bf g}_i - \hat{{\bf g}}_i  ||^2_2\right)
 = {\rm Tr}\left[ {\bf C}_{0,G_i} ({\bf w}_{G_i} - \hat{{\bf w}}_{G_i})({\bf w}_{G_i} - \hat{{\bf w}}_{G_i})^H \right] \\
& = {\rm Tr}\left[ {\bf C}_{0,G_i} ({\bf w}_{H} - {\bf T}_i {\bf y}_i)({\bf w}_{H} - {\bf T}_i {\bf y}_i)^H \right],~i=1,2
\end{split}
\end{equation}
where ${\bf C}_{0,G_i} \triangleq {\bf C}^H_{t,G} {\bf C}_{t,G} \otimes {\bf C}^H_{r,G_i} {\bf C}_{r,G_i} $ and
\begin{equation}\label{III-722}
\begin{split}
{\bf T}_i = \mathfrak{R}_{w_{G_i} y_i} \mathfrak{R}^{-1}_{y_i y_i}.
\end{split}
\end{equation}
In \eqref{III-722},
\begin{equation}\label{III-73}\nonumber
\begin{split}
\mathfrak{R}_{w_{G_i} y_i} =& {\cal E}({\bf w}_{G_i} {\bf y}^H_i)
 =   {\bf C}^H_{t,G} {\bf S}^*_R \otimes {\bf C}^H_{r,G_i}\\
\mathfrak{R}_{y_i y_i} =& {\cal E}({\bf y}_i {\bf y}^H_i)
 = \left({\bf S}^T_i {\bf C}_{t,G}\otimes {\bf C}_{r,G_i} \right) \left({\bf S}^T_R {\bf C}_{t,G}\otimes {\bf C}_{r,G_i} \right)^H + {\bf K}_i
 = {\bf S}^T_R {\bf C}_{t,G} {\bf C}^H_{t,G} {\bf S}^*_R \otimes {\bf C}_{r,G_i} {\bf C}^H_{r,G_i} + {\bf K}_i.
\end{split}
\end{equation}
By substituting ${\bf T}_i$ into \eqref{III-72}, we obtain %\textcolor{red}{(! $e_i $ has been defined in \eqref{III-72})}
\begin{equation}\label{III-74}
\begin{split}
e_i  & = {\rm Tr}\left[{\bf C}_{0,G_i} \left({\bf I} + \left({\bf S}^T_R {\bf C}_{t,G}\otimes {\bf C}_{r,G_i} \right)^H {\bf K}^{-1}_i
 \left({\bf S}^T_R {\bf C}_{t,G}\otimes {\bf C}_{r,G_i} \right)  \right )^{-1} \right].
\end{split}
\end{equation}

\vspace{+5pt}
\section{Training Sequence Design for MAC Phase}
In this section, the design of the training sequences for the MAC phase are analyzed. Namely, we shall optimize the training sequences ${\bf S}_1$ and ${\bf S}_2$ subject to two source power constraints to minimize the total estimation MSE, i.e., $e_R$ in \eqref{III-6}.
The corresponding training sequence optimization problem can be formulated as
%According to $e_R$ derived in \eqref{III-7}, the training sequence design optimization problem can be formulated as
\begin{equation}\label{IV-1}
\begin{split}
 \min_{{\bf S}_1, {\bf S}_2} ~~& {\rm Tr}\left[ {\bf C}_{0,H} \left({\bf I} + \left({\bf S}^T {\bf C}_{t,H}\otimes {\bf C}_{r,H} \right)^H {\bf K}^{-1}_R
 \left({\bf S}^T {\bf C}_{t,H}\otimes {\bf C}_{r,H} \right)  \right )^{-1} \right]  \\
{\rm s.t.}~~&  {\rm Tr}({\bf S}_i {\bf S}^H_i) \leq \tau_i,~i=1,2.
\end{split}
\end{equation}

Before solving \eqref{IV-1}, we first introduce the following lemma that deals with the minimum length of the training sequence ${\bf S}$, $L_S$.
\vspace{-7pt}
%\textbf{Theorem} 1:
\begin{lemma}\label{MAC_length}
To achieve an arbitrary small MSE with infinite power at the source nodes, the minimum length of the source training sequence should be set to $L_S = N_1+N_2$. Otherwise, even with infinite power at the source nodes, the total MSE is lower bounded by $\sum^M_{n=1}\sigma_{r,H,n}\sum^{N_1+N_2}_{m=L_S+1}{\sigma_{t,H,m}}$
{with $\sigma_{t,H,m}$ being the $m$-th element of ${\bm \lambda}({\bf Z}_{r,H})$} and ${\bf Z}_{t,H}={\rm Blkdiag}({\bf Z}_{t,H_1}, {\bf Z}_{t,H_2})$. Moreover, for ${\bf K}_{q,R} = q {\bf I}$ and any power constraint at the source node, if the optimal solution of ${\bf S}$ in \eqref{IV-1} has a rank $r$, the minimum length of source training sequence can be set to $L_S = r$.
\end{lemma}
\vspace{-7pt}
\begin{proof}
See Appendix~\ref{prof_theorem1}
\end{proof}

%\begin{proof}
%Please refer to Appendix~\ref{prof_theorem1}.
%\end{proof}

With the minimum length of the training sequences determined, we seek to solve the non-convex optimization problem in \eqref{IV-1} with respect to ${\bf S}_1$ and ${\bf S}_2$. Although the objective function in \eqref{IV-1} has a similar form to that of point-to-point systems, there are two power constraints in \eqref{IV-1} that make the problem of solving this non-convex optimization problem more difficult than that of point-to-point systems in \cite{Biguesh_tsp2009, Bjornson2010}. In order to proceed, we first note that $e_R$ in \eqref{III-7} can be obtained by substituting \eqref{III-4} into {\eqref{III-6}}. Thus, to make the problem tractable, we propose an iterative algorithm, which decouples the primal problem into two sub-problems and solves each of them in an alternating approach.
Let us rewrite \eqref{III-7} into the following form
\begin{equation}\label{IV-2}
\begin{split}
\tilde{e}_R = & {\rm Tr}\left[ {\bf C}_{0,H} ({\bf w}_{H} - {\bf T}_R {\bf y}_R)({\bf w}_{H} - {\bf T}_R {\bf y}_R)^H \right] \\
 = & {\rm Tr}\left[{\bf C}_{0,H} -  \left({\bf S}^T {\bf C}_{t,H}\otimes {\bf C}_{r,H} \right)^H {\bf T}^H_R {\bf C}^H_{0,H} -
  {\bf C}_{0,H}{\bf T}_R \left({\bf S}^T {\bf C}_{t,H}\otimes {\bf C}_{r,H} \right)
  \right. \\
 &+ \left. {\bf C}_{0,H} {\bf T}_R \left({\bf S}^T {\bf C}_{t,H}\otimes {\bf C}_{r,H} \right)
 \times \left({\bf S}^T {\bf C}_{t,H}\otimes {\bf C}_{r,H} \right)^H {\bf T}^H_R
 + {\bf C}_{0,H} {\bf T}_R {\bf K}_R {\bf T}^H_R\right].
\end{split}
\end{equation}
Then, the optimization problem in \eqref{IV-1} is equivalent to
\begin{equation}\label{IV-22-a}
\begin{split}
& \min_{{\bf T}_R, {\bf S}_1, {\bf S}_2}~~   \tilde{e}_R  \\
& {\rm s.t.} ~~  {\rm Tr}({\bf S}_i {\bf S}^H_i) \leq \tau_i,~i=1,2.
\end{split}
\end{equation}

In the first subproblem, we intend to optimize the LMMSE estimator matrix ${\bf T}_R$ for a given ${\bf S}_1$ and ${\bf S}_2$. Since ${\bf T}_R$ is not related to the the power constraint, the problem simplifies to an unconstrained optimization problem given by
\begin{equation}\label{IV-222}
\begin{split}
 \min_{{\bf T}_R}~~  \tilde{e}_R.
\end{split}
\end{equation}
%It is easy to verify the the optimal ${\bf T}_R$ is given in \eqref{III-4}
Given that \eqref{IV-222} is convex with respect to ${\bf T}_R$, by setting its gradient to zero, it can be shown that the optimal ${\bf T}_R$ is equal to \eqref{III-4}.
%(! note we need \eqref{IV-222} since in \eqref{III-4}, ${\bf T}_R$ is derived for estimating ${\bf W}_H$. In fact in \eqref{IV-22}, we want to estimate $\bf H$.) }

In the second subproblem, the training sequences ${\bf S}_1$ and ${\bf S}_2$ need to be optimized for a given ${\bf T}_R$ by solving the following optimization problem
\begin{equation}\label{IV-2222}
\begin{split}
 & \min_{{\bf S}_1, {\bf S}_2}~~ \tilde{e}_R \\
& {\rm s.t.}~~  {\rm Tr}({\bf S}_i {\bf S}^H_i) \leq \tau_i,~i=1,2.
\end{split}
\end{equation}
Next, it is shown that the optimization problem in \eqref{IV-2222} can be transformed into a convex quadratically constrained quadratic programable (QCQP) problem \cite{Boyd2004}. To achieve this goal, we first reformulate the last term in \eqref{IV-2} as
\begin{equation}\label{IV-3}
\begin{split}
& {\rm Tr}\left[ {\bf C}_{0,H}{\bf T}_R \left({\bf S}^T {\bf C}_{t,H}\otimes {\bf C}_{r,H} \right)\left({\bf S}^T {\bf C}_{t,H}\otimes {\bf C}_{r,H} \right)^H {\bf T}^H_R \right] \\
\overset{(a)}{=}& {\rm Tr}\left[ {\bf T}^H_R {\bf C}_{0,H}{\bf T}_R ({\bf S}^T \otimes {\bf I}) ({\bf C}_{t,H} {\bf C}^H_{t,H} \otimes {\bf C}_{r,H} {\bf C}^H_{r,H}) ({\bf S}^* \otimes {\bf I}) \right]  \\
 \overset{(b)}{=} & {\rm vec}({\bf S} \otimes {\bf I})^H \left( {\bf T}^H_R {\bf C}_{0,H} {\bf T}_R \otimes {\bf C}^T_{tr}   \right) {\rm vec}({\bf S} \otimes {\bf I}) \\
 \overset{(c)}{=}& {\bf s}^H {\bf E}^H \left( {\bf T}^H_R {\bf C}_{0,H} {\bf T}_R \otimes {\bf C}^T_{tr}   \right) {\bf E} {\bf s},
\end{split}
\end{equation}
where ${\bf C}_{tr} \triangleq {\bf C}_{t,H} {\bf C}^H_{t,H} \otimes {\bf C}_{r,H} {\bf C}^H_{r,H}$, ${\bf s} \triangleq {\rm vec}({\bf S} )$, ${\bf E} \triangleq {\rm Blkdiag}(\tilde{{\bf E}}_{(1)},\tilde{{\bf E}}_{(2)},\cdots, \tilde{{\bf E}}_{(L_S)} )$, $\tilde{{\bf E}}_{(i)} = \tilde{{\bf E}}$, $\tilde{{\bf E}} \triangleq \left[\bar{{\bf E}}_{(1)};\bar{{\bf E}}_{(2)};\cdots; \bar{{\bf E}}_{(M)} \right]$, $\bar{{\bf E}}_{(i)} \triangleq {\rm Blkdiag}(\underbrace{{\bf e}_i,{\bf e}_i,\cdots, {\bf e}_i}_{N_1+N_2 ~\text{elements}} )$, and ${\bf e}_i \triangleq [0,0,\cdots,\underbrace{1}_{i-\text{th element}},\cdots,0]^T$.  In \eqref{IV-3}, Eq. $(a)$ is obtained by using the circular property ${\rm Tr}\{{\bf A}{\bf B}\}={\rm Tr}\{{\bf B}{\bf A}\}$ and the matrix identity
\begin{equation}\label{IV-4}
\begin{split}
({\bf A}\otimes {\bf B})({\bf C}\otimes {\bf D}) ={\bf A} {\bf C} \otimes {\bf B} {\bf D},
\end{split}
\end{equation}
Eq. $(b)$ is obtained by using the identity
\begin{equation}\label{IV-5}
\begin{split}
{\rm Tr}({\bf A} {\bf B} {\bf C} {\bf D}) = {\rm vec}({\bf D})^T ({\bf A} \otimes {\bf C}^T)  {\rm vec}({\bf B}^T)~{\rm and}~
({\bf A}\otimes {\bf B} )^H &= {\bf A}^H \otimes {\bf B}^H,
\end{split}
\end{equation}
and Eq. $(c)$ is obtained by using ${\rm vec}({\bf S} \otimes {\bf I}) = {\bf E} {\bf s}$. Similarly, the term ${\rm Tr}\left[{\bf C}_{0,H}{\bf T}_R ({\bf S}^T {\bf C}_{t,H}\otimes {\bf C}_{r,H} )\right]$ can be expressed as
\begin{equation}\label{IV-7}
\begin{split}
 {\rm Tr}\left[{\bf C}_{0,H}{\bf T}_R ({\bf S}^T {\bf C}_{t,H}\otimes {\bf C}_{r,H} )\right]
  = &  {\rm Tr}\left[ ({\bf S}\otimes {\bf I} )^T ({\bf C}_{t,H}\otimes {\bf C}_{r,H}) {\bf C}_{0,H}{\bf T}_R \right] \\
 = & {\rm vec}({\bf S}\otimes {\bf I})^T  {\rm vec}({\bf C}_T)   \\
 = & {\rm vec}({\bf C}_T)^T {\bf E} {\bf s},
\end{split}
\end{equation}
where ${\bf C}_T \triangleq ({\bf C}_{t,H}\otimes {\bf C}_{r,H}) {\bf C}_{0,H}{\bf T}_R$. To obtain \eqref{IV-7}, we use the fact
\begin{equation}\label{IV-8}
\begin{split}
{\rm Tr}({\bf A}^T{\bf B}) = {\rm vec}({\bf A})^T {\rm vec}({\bf B}).
\end{split}
\end{equation}
The source power constrain in \eqref{IV-2222} can be rewritten as
\begin{equation}\label{IV-9}
\begin{split}
{\rm Tr}({\bf S}_i {\bf S}^H_i) = {\rm Tr}({\bf E}_i {\bf S} {\bf S}^H),
\end{split}
\end{equation}
where ${\bf E}_1 \triangleq {\rm Blkdiag}({\bf I}_{N_1}, {\bf 0}_{N_2\times N_2})$ and ${\bf E}_2 \triangleq {\rm Blkdiag}({\bf 0}_{N_1\times N_1}, {\bf I}_{N_2})$.
Based on the property that ${\rm Tr}({\bf A} {\bf B} {\bf C} {\bf D}) = {\rm vec}({\bf D}^T)^T ({\bf C}^T \otimes {\bf A})  {\rm vec}({\bf B})$, Eq. \eqref{IV-9} can be further modified as
\begin{equation}\label{IV-11}
\begin{split}
{\rm Tr}({\bf S}_i {\bf S}^H_i) =  {\bf s}^H ({\bf I} \otimes {\bf E}_i) {\bf s}.
\end{split}
\end{equation}
According to \eqref{IV-3}, \eqref{IV-7}, and \eqref{IV-11}, the optimization problem in \eqref{IV-2222} can be transformed into
\begin{equation}\label{IV-12}
\begin{split}
  \min_{{\bf s}}&~~ {\bf s}^H {\bf E}^H \left( {\bf T}^H_R {\bf C}_{0,H}{\bf T}_R \otimes {\bf C}^T_{tr}   \right) {\bf E} {\bf s} - 2 \Re({\rm vec}({\bf C}_T)^T {\bf E} {\bf s})\\
 {\rm s.t.}& ~~ {\bf s}^H ({\bf I} \otimes {\bf E}_i) {\bf s} \leq \tau_i,~i=1,2.
\end{split}
\end{equation}
Since both ${\bf E}^H \left( {\bf T}^H_R {\bf T}_R \otimes {\bf C}^T_{tr}   \right) {\bf E}$ and ${\bf I} \otimes {\bf E}_i$ are positive semidefinite matrices, we conclude that the optimization problem in \eqref{IV-12} is a convex QCQP problem, which can be easily solved by applying the available software package in for example \cite{CVX}.
%\textcolor{red}{(! Your concerns may be right since I just do some linear algebra transformations. However, from the original problem \eqref{IV-2}, it is not easy to see that are convex QCQP problem. Hence i think here it is fine. My several papers also write like this. No reviewer commented on this. )}

In summary, we outline the proposed iterative training design algorithm as follows:

%\vspace{-0.2cm}
\vspace{-9pt}
\hrulefill
\par
{\footnotesize
\textbf{Algorithm 1}
\begin{itemize}
\item \textbf{Initialize} ${\bf S}_1$, ${\bf S}_2$
\item \textbf{Repeat}
\begin{itemize}
\item Update the LMMSE estimator matrix ${\bf T}_R$ using \eqref{III-4} for fixed ${\bf S}_1$ and ${\bf S}_2$;
\item For fixed ${\bf T}_R$, solve the convex QCQP problem in \eqref{IV-12} to get the optimal ${\bf S}_1$ and ${\bf S}_2$;
\end{itemize}
\item \textbf{Until} The difference between the MSE from one iteration to another is smaller than a certain predetermined threshold.
\end{itemize}}
%\vspace{-0.5cm}
\vspace{-9pt}
\hrulefill
\vspace{-7pt}
%\textbf{Theorem 2}:
\begin{theorem}\label{MAC_conveg}
The proposed iterative precoding design in Algorithm 1 is convergent and the limit point of the iteration is a stationary point of
\eqref{IV-22-a}.
\end{theorem}
\vspace{-7pt}
\begin{proof}
Since in the proposed iterative algorithm, the solution for each subproblem is optimal, the total MSE is {not increased} after each iteration. Meanwhile, the total MSE is lower bounded  by zero. Hence, the proposed algorithm is convergent.
It further means that there must exist a
limit point, denoted as $ \bar{{\bf X}} = \left\{{\bf \bar{T}}_R, {\bf \bar{S}}_i, i=1,2 \right\}$
%\textcolor{red}{(! ${\bf \bar{T}}_R, {\bf \bar{S}}_i$ are just the converged point.)},
after the convergence. At the limit point, the solutions will not change if we continue the iteration. Otherwise, the total MSE can be further decreased, which contradicts the assumption of convergence.

Since at the limit point $ \bar{{\bf X}}$, ${\bf \bar{T}}_R$ is the local minimizer of subproblem \eqref{IV-222}, it can be concluded that ${\bf \bar{T}}_R$ satisfies the following Karush-Kuhn-Tucker (KKT) condition \cite{Boyd2004}
\begin{equation}\label{IV-13}
\begin{split}
\frac{\partial \tilde{e}_R({\bf \bar{S}}_1,{\bf \bar{S}}_2)}{\partial {\bf T}_R}|_{{\bf T}_R = {\bf \bar{T}}_R} = {\bf 0},
\end{split}
\end{equation}
where $\tilde{e}_R({\bf \bar{S}}_1,{\bf \bar{S}}_2)$ denotes the function $e_R$ with ${\bf S}_1$ and ${\bf S}_2$ being evaluated at ${\bf \bar{S}}_1$ and ${\bf \bar{S}}_2$, respectively. Similarly, ${\bf \bar{S}}_i$, for $i=1,2$, is the local minimizer of subproblem \eqref{IV-2222}, which satisfies the following KKT conditions
\begin{equation}\label{IV-14}
\begin{split}
\frac{\partial \tilde{e}_R({\bf \bar{T}}_R)}{\partial {\bf T}_R}|_{{\bf S}_i = {\bf \bar{S}}_i} = {\bf 0},~
\lambda_i \left({\rm Tr}(\bar{{\bf S}}_i \bar{{\bf S}}^H_i)-\tau_i \right) = 0, ~~i=1,2
\end{split}
\end{equation}
where $\lambda_i$ is the lagrangian multiplier associated with the source power constraint. By summing up the KKT conditions given in \eqref{IV-13} and \eqref{IV-14}, it can be concluded that the limit point $ \bar{{\bf X}}$ satisfies the KKT conditions of the primal problem in \eqref{IV-22-a}, which further means that $ \bar{{\bf X}}$ is a stationary point of \eqref{IV-22-a}.
%\textcolor{red}{(! Hani, do not worry about this proof. The similar proof has been used in Prof. Zhiquan Luo's papers. He is very good at optimization. The similar proof has also been given in one my TSP paper. It is indeed right.)}
\end{proof}

To this point, it is shown that the joint source training design can be solved via \emph{Algorithm 1}. In the following, we illustrate that for some special cases, the optimal solution of \eqref{IV-1} can be obtained in closed-form.

\subsection{When ${\bf K}_{r,R} = {\bf Z}_{r,H}$}

We first consider the case with ${\bf K}_{r,R} = {\bf Z}_{r,H}$. This case is applicable in the scenario where the disturbance is dominated by the interference from neighboring users as shown in \cite{Yong_tsp2007}. As summarized in Section II, this scenario corresponds to the interference-limited case. Accordingly, the LMMSE estimator given in \eqref{III-4} can be rewritten as
\begin{equation}\label{III-15}\nonumber
\begin{split}
 {\bf T}_R
= & \left[{\bf C}^H_{t,H} {\bf S}^* \otimes {\bf C}^H_{r,H} \right]
 \left[{\bf S}^T {\bf C}_{t,H} {\bf C}^H_{t,H} {\bf S}^* \otimes {\bf C}_{r,H} {\bf C}^H_{r,H} + {\bf K}_R \right]^{-1} \\
 = &\left[ {\bf C}^H_{t,H} {\bf S}^* \otimes {\bf C}^H_{r,H}\right]
 \left[ \left( {\bf S}^T {\bf C}_{t,H} {\bf C}^H_{t,H} {\bf S}^* + {\bf K}_{q,R} \right)^{-1} \otimes {\bf Z}^{-1}_{r,H}\right] \\
 = & \underbrace{{\bf C}^H_{t,H} {\bf S}^* \left( {\bf S}^T {\bf C}_{t,H} {\bf C}^H_{t,H} {\bf S}^* + {\bf K}_{q,R} \right)^{-1}}_{\triangleq {\bf T}_{R,1}} \otimes {\bf C}^{-1}_{r,H}, \\
\end{split}
\end{equation}
which further leads to
\begin{equation}\label{III-15}\nonumber
\begin{split}
{\bf T}^H_R {\bf C}_{0,H} {\bf T}_R  =  ({\bf T}_{R,1}  \otimes {\bf C}^{-1}_{r,H})^H ({\bf C}^H_{t,H} {\bf C}_{t,H} \otimes {\bf C}^H_{r,H} {\bf C}_{r,H})  \times ({\bf T}_{R,1}  \otimes {\bf C}^{-1}_{r,H})
=   {\bf T}^H_{R,1}{\bf C}^H_{t,H} {\bf C}_{t,H} {\bf T}_{R,1} \otimes {\bf I},
\end{split}
\end{equation}
%Hence the last term in \eqref{IV-2} can be
and
\begin{equation}\label{IV-16}
\begin{split}
& {\rm Tr}\left[{\bf C}_{0,H} {\bf T}_R \left({\bf S}^T {\bf C}_{t,H}\otimes {\bf C}_{r,H} \right)\left({\bf S}^T {\bf C}_{t,H}\otimes {\bf C}_{r,H} \right)^H {\bf T}^H_R \right] \\
= & {\rm Tr}\left[\left({\bf S}^T {\bf C}_{t,H}{\bf C}^H_{t,H}{\bf S}^* \otimes {\bf C}_{r,H} {\bf C}^H_{r,H}  \right)
 \left( {\bf T}^H_{R,1}{\bf C}^H_{t,H} {\bf C}_{t,H} {\bf T}_{R,1} \otimes {\bf I} \right)\right] \\
\overset{(a)}{=}& {\rm Tr}({\bf Z}_{r,H}) {\rm Tr}\left[{\bf S}^T {\bf C}_{t,H}{\bf C}^H_{t,H}{\bf S}^* {\bf T}^H_{R,1}{\bf C}^H_{t,H} {\bf C}_{t,H}{\bf T}_{R,1}\right] \\
\overset{(b)}{=}& {\rm Tr}({\bf Z}_{r,H}) \sum^{2}_{i=1}{\rm Tr}\left({\bf S}^T_i {\bf Z}_{t,H_i} {\bf S}^*_i {\bf T}^H_{R,1}{\bf C}^H_{t,H}{\bf C}_{t,H}{\bf T}_{R,1} \right).
\end{split}
\end{equation}
In \eqref{IV-16}, Eq. $(a)$ is obtained by using the identity
%\begin{equation}\label{IV-17}\nonumber
%\begin{split}
${\rm Tr}({\bf A}\otimes {\bf B}) = {\rm Tr}({\bf A}) {\rm Tr}({\bf B})$,
%\end{split}
%\end{equation}
and Eq. $(b)$ is derived based on the fact that ${\bf C}_{t,H}$ is a block diagonal matrix as shown in \eqref{III-1}.
In addition,  the term ${\rm Tr}\left[{\bf C}_{0,H}{\bf T}_R ({\bf S}^T {\bf C}_{t,H}\otimes {\bf C}_{r,H} )\right]$ can be reexpressed as
\begin{equation}\label{IV-18}
\begin{split}
 {\rm Tr}\left[{\bf C}_{0,H}{\bf T}_R ({\bf S}^T {\bf C}_{t,H}\otimes {\bf C}_{r,H} )\right]
& =  {\rm Tr}\left[ ({\bf S}^T {\bf C}_{t,H} {\bf C}^H_{t,H} {\bf C}_{t,H} {\bf T}_{R,1} ) \otimes {\bf C}_{r,H} {\bf C}^H_{r,H} \right] \\
& =  {\rm Tr}({\bf Z}_{r,H}) \sum^{2}_{i=1} {\rm Tr}\left( {\bf S}^T_i {\bf Z}_{t,H_i} {\bf C}_{t,H_i} {\bf T}_{R,1,i}\right),
\end{split}
\end{equation}
where ${\bf T}_{R,1,1} \triangleq {\bf T}_{R,1}(1:N_1,:)$ and ${\bf T}_{R,1,2} \triangleq {\bf T}_{R,1}(N_1+1:N_1+N_2,:)$.
Based on \eqref{IV-16} and \eqref{IV-18}, \eqref{IV-2222} is equivalent to the following optimization problem
\begin{equation}\label{IV-19}
\begin{split}
  \min_{{\bf S}_1, {\bf S}_2}~~& \sum^{2}_{i=1}\left\{ {\rm Tr}\left({\bf S}^T_i {\bf Z}_{t,H_i} {\bf S}^*_i {\bf T}^H_{R,1}{\bf C}^H_{t,H}{\bf C}_{t,H}{\bf T}_{R,1} \right) \right. \\
  & \left. - {\rm Tr}\left( {\bf S}^T_i {\bf Z}_{t,H_i} {\bf C}_{t,H_i} {\bf T}_{R,1,i}\right) - {\rm Tr}\left( {\bf T}^H_{R,1,i} {\bf C}^H_{t,H_i} {\bf Z}^H_{t,H_i} {\bf S}^*_i   \right)\right\}\\
 {\rm s.t.}& ~~~ {\rm Tr}({\bf S}_i {\bf S}^H_i) \leq \tau_i,~i=1,2.
\end{split}
\end{equation}
We note that compared to \eqref{IV-2222}, \eqref{IV-19} has a simpler form and as shown below can be solved in closed-form via the KKT conditions.

To proceed, the lagrangian function of \eqref{IV-19} is first derived as
\begin{equation}\label{IV-20}\nonumber
\begin{split}
  \mathcal{L} = &  \sum^{2}_{i=1}\left\{ {\rm Tr}\left({\bf S}^T_i {\bf Z}_{t,H_i} {\bf S}^*_i {\bf T}^H_{R,1}{\bf C}^H_{t,H}{\bf C}_{t,H}{\bf T}_{R,1} \right)
   - {\rm Tr}\left( {\bf S}^T_i {\bf Z}_{t,H_i} {\bf C}_{t,H_i} {\bf T}_{R,1,i}\right) - {\rm Tr}\left( {\bf T}^H_{R,1,i} {\bf C}^H_{t,H_i}{\bf Z}^H_{t,H_i} {\bf S}^*_i   \right)\right\}\\
 & +\sum^2_{i=1} \lambda_i \left[ {\rm Tr}({\bf S}_i {\bf S}^H_i) - \tau_i \right],
\end{split}
\end{equation}
where $\lambda_i$ is the lagrangian multiplier associated with the power constraint at the source $S_i$. The KKT conditions for \eqref{IV-19} can be derived as follows \cite{Boyd2004}
\begin{subequations}
\begin{align}\label{IV-21}
  \frac{\partial \mathcal{L}}{\partial {\bf S}^*_i} & =  \left({\bf T}^H_{R,1} {\bf C}^H_{t,H}{\bf C}_{t,H} {\bf T}_{R,1} {\bf S}^T_i {\bf Z}_{t,H_i}\right)^T -
   \left( {\bf T}^H_{R,1,i} {\bf C}^H_{t,H_i}{\bf Z}^H_{t,H_i} \right)^T + \lambda_i {\bf S}_i
   = {\bf 0}, \\
\label{IV-22}
 \lambda_i \left[ {\rm Tr}({\bf S}_i {\bf S}^H_i) - \tau_i \right] &= 0, \\
\label{IV-23}
  {\rm Tr}({\bf S}_i {\bf S}^H_i)&\leq  \tau_i, ~~i=1,2.
\end{align}
\end{subequations}
Based on the KKT conditions shown in \eqref{IV-21} and by using \eqref{III-1-1}, the optimal ${\bf s}_i \triangleq{\rm vec}({\bf S}_i)$ can be obtained as
\begin{equation}\label{IV-24}
\begin{split}
  {\bf s}_i = \left[{\bf X}_{s,1} \otimes {\bf X}_{s,2,i} + \lambda_i {\bf I}\right]^{-1} {\bf x}_{s,3,i},
\end{split}
\end{equation}
where ${\bf X}_{s,1}\triangleq{\bf T}^H_{R,1} {\bf C}^H_{t,H}{\bf C}_{t,H} {\bf T}_{R,1}$, ${\bf X}_{s,2,i} \triangleq {\bf Z}^T_{t,H_i} $, and ${\bf x}_{s,3,i}\triangleq{\rm vec}\big(\big( {\bf T}^H_{R,1,i} {\bf C}^H_{t,H_i}{\bf Z}^H_{t,H_i}\big)^T \big)$.
%To obtain \eqref{IV-24}, we use the rule given in \eqref{III-1-1}.
The optimal $\lambda_i$ in \eqref{IV-24} can be zero or should be chosen to activate the power constraint in \eqref{IV-23}. For the case where $\lambda_i \neq 0$, the following lemma is introduced.
\vspace{-7pt}
%\textbf{Lemma 1}:
\begin{lemma}\label{MAC_sec_A}
The function $g(\lambda_i)={\rm Tr} \left\{ {\bf S}_i {\bf S}^H_i \right\}= {\rm Tr} \left\{ {\bf s}_i {\bf s}^H_i \right\}$, with ${\bf s}_i$ defined above \eqref{IV-24}, is monotonically decreasing with respect to $\lambda_i$ and the optimal $\lambda_i$ is upper-bounded by ${\displaystyle\sqrt{\frac{\sigma_{s,3,i}}{\tau_i}}-\sigma_{s,\text{min},i}}$. Here, $\sigma_{s,3,i}$ denotes the smallest eigenvalue of ${\bf X}_{s,1} \otimes {\bf X}_{s,2,i}$ and $\sigma_{s,\text{min},i}=||{\bf x}_{s,3,i}||^2_2$.
\end{lemma}
\vspace{-7pt}
\begin{proof}
See Appendix~\ref{prof_lemma1}.
\end{proof}

By applying \emph{Lemma~\ref{MAC_sec_A}}, the optimal $\lambda_i$ that meets the condition $ {\rm Tr} \left\{{\bf S}_i {\bf S}^H_i \right\} = \tau_i$ can be readily obtained via the bisection search algorithm.

\subsection{When ${\bf K}_{q,R} = q {\bf I}$}\label{sec:4-b}
This scenario corresponds to the practical case, where the disturbance consists of both the additive white Gaussian noise and the temporally uncorrelated interference. This is referred to as the additive noise-limited and spatially uncorrelated interference scenario in Section II. Similar to the derivation steps in Appendix~\ref{prof_theorem1}, the total MSE can be derived as
\begin{equation}\label{IV-26}\nonumber
\begin{split}
e_R  = & {\rm Tr}\left[ {\bf C}_{0,H} \left({\bf I} + \left({\bf S}^T {\bf C}_{t,H}\otimes {\bf C}_{r,H} \right)^H {\bf K}^{-1}_R
\left({\bf S}^T {\bf C}_{t,H}\otimes {\bf C}_{r,H} \right)  \right )^{-1} \right] \\
= &  {\rm Tr}\left[ {\bf C}_{0,H} \left({\bf I} + \frac{1}{q}{\bf C}^H_{t,H} {\bf S}^* {\bf S}^T {\bf C}_{t,H} \otimes {\bf C}^H_{r,H} {\bf K}^{-1}_{r,R} {\bf C}_{r,H} \right )^{-1} \right] \\
= & \sum^{M}_{n=1} \sigma_{r,H,n} {\rm Tr}\left[ \left({\bf Z}^{-1}_{t,H} + \alpha_n  {\bf S}^*  {\bf S}^T  \right )^{-1} \right],
\end{split}
\end{equation}
where $\displaystyle{\alpha_n = \frac{\sigma_{r,H,n}}{q \delta_{r,R,n}}}$. Subsequently, the optimization problem in \eqref{IV-1} can be rewritten as
\begin{equation}\label{IV-27}
\begin{split}
 \min_{{\bf S}_1, {\bf S}_2}~~& \sum^{M}_{n=1} \sigma_{r,H,n} {\rm Tr}\left[ \left({\bf Z}^{-1}_{t,H} + \alpha_n  {\bf S}^*  {\bf S}^T  \right )^{-1} \right] \\
 {\rm s.t.} ~~& {\rm Tr}({\bf E}_i{\bf S}^* {\bf S}^T) \leq \tau_i,~i=1,2
\end{split}
\end{equation}
where ${\bf E}_i$ is defined in \eqref{IV-9}. Although the optimization problem in \eqref{IV-27} is non-convex with respect to ${\bf S}_i$, it is noted that one may optimize \eqref{IV-27} with respect to the positive semidefinite matrix ${\bf S}^* {\bf S}^T$ instead of the training sequence ${\bf S}_i$. Accordingly, after solving for the optimum ${\bf S}^* {\bf S}^T$, the solution can be decomposed to obtain $S_i$. This approach is preferable since in \eqref{IV-27}, the objective function and the constraint both depend on ${\bf S}^* {\bf S}^T$ and not ${\bf S}_i$. Hence, by defining ${\bf Q}_S \triangleq {\bf S}^* {\bf S}^T$, the following equivalent problem can be obtained
\begin{equation}\label{IV-28}
\begin{split}
 \min_{{\bf Q}_S \succeq {\bf 0}}~~ & \sum^{M}_{n=1} \sigma_{r,H,n} {\rm Tr}\left[ \left({\bf Z}^{-1}_{t,H} + \alpha_n  {\bf Q}_S  \right )^{-1} \right]  \\
{\rm s.t.} ~~& {\rm Tr}({\bf E}_i{\bf Q}_S) \leq \tau_i,~i=1,2.
\end{split}
\end{equation}
\vspace{-7pt}
\begin{theorem}\label{MAC_convex}
The optimization problem in \eqref{IV-28} is convex with respect to the positive semidefinite matrix ${\bf Q}_S$.
\end{theorem}
\vspace{-7pt}
\begin{proof}
See Appendix~\ref{prof_theorem3}.
\end{proof}

\noindent Next, it is shown that the optimization problem in \eqref{IV-28} can be solved by transforming it into a semidefinite programming (SDP) problem. By introducing the variables ${\bf X}_n$, the problem in \eqref{IV-28} can be rewritten in an equivalent form as
\begin{equation}\label{IV-29}
\begin{split}
 \min_{{\bf Q}_S \succeq {\bf 0},{\bf X}_n}~~ &\sum^{M}_{n=1} \sigma_{r,H,n}{\rm Tr}\left({\bf X}_n \right)  \\
{\rm s.t.} ~~& {\rm Tr}({\bf E}_i{\bf Q}_S) \leq \tau_i,~i=1,2 \\
&  \left({\bf Z}^{-1}_{t,H} + \alpha_n  {\bf Q}_S \right )^{-1} \preceq {\bf X}_n,\forall n
\end{split}
\end{equation}
By using the Schur complement, \eqref{IV-29} can be further transformed into the following SDP problem
\begin{equation}\label{IV-30}
\begin{split}
 \min_{{\bf Q}_S \succeq {\bf 0}, {\bf X}_n }~~ &\sum^{M}_{n=1}\sigma_{r,R,n} {\rm Tr}\left({\bf X}_n \right)  \\
{\rm s.t.} ~~& {\rm Tr}({\bf E}_i{\bf Q}_S) \leq \tau_i,~i=1,2 \\
%&  {\bf Q}_S \succeq {\bf 0} \\
&  \left[
     \begin{array}{cc}
       {\bf Z}^{-1}_{t,H} + \alpha_n  {\bf Q}_S  & {\bf I} \\
       {\bf I} & {\bf X}_n\\
     \end{array}
   \right]
 \succeq {\bf 0},~\forall n
\end{split}
\end{equation}
By solving the SDP problem in \eqref{IV-30}, the optimal solution to the optimization problem in \eqref{IV-28} can be obtained.
However, this numerical method of solving this optimization problem has a relatively high computational complexity. As such to obtain the optimal structure of ${\bf S}_i$ and gain a better understanding of the optimization in \eqref{IV-27} the following theorem is introduced.
\vspace{-7pt}
%\textbf{Theorem 4}:
\begin{theorem}\label{MAC_sec_B}
With the minimum training sequence length requirement $L_S \geq N_1+N_2$, the optimal training sequence ${\bf S}_i$ in \eqref{IV-27} should satisfy the condition ${\bf S}^*_1{\bf S}^T_2 =0$. In addition, the optimal ${\bf S}_i$ has a form of ${\bf S}_i = {\bf U}^*_{t, H_i} {\bm \Sigma}_{s_i} {\bf V}^H_{s_i}$, where ${\bf V}_{s_i}$ is chosen such that ${\bf V}^H_{s_1} {\bf V}_{s_2} = {\bf 0}$ and ${\bm \Sigma}_{s_i}$ is a diagonal eigenvalue matrix with $[{\bm \Sigma}_{s_i}]_{n,n} ={\sigma}_{s_i,n} $. ${\bm \Sigma}_{s_i}$ can be obtained by solving the following water-filling problem
\begin{equation}\label{App-IV-10}
\begin{split}
\sum^{M}_{n=1} \frac{\alpha_n \sigma_{r,H,n} \sigma^2_{t,H_i,m}}{( 1+\alpha_n \sigma_{t,H_i,m} \sigma^2_{s_i,m})^2} = \lambda_i.
\end{split}
\end{equation}
In \eqref{App-IV-10}, the optimal $\lambda_i$ should be selected such that $\sum^{N_i}_{m=1}\sigma^2_{s_i,m} = \tau_i$.
\end{theorem}
\begin{proof}
See Appendix~\ref{prof_theorem4}.
\end{proof}
%\newtheorem{remark}{Remark}
%\begin{remark}\label{rem:00}

\noindent \emph{Remark 1:} the optimal $\lambda_i$ in \eqref{App-IV-10} can be found via the bisection search algorithm and the optimal $\lambda_i$ is bounded by $\big(0, \max_{m}  \sum^{M}_{n=1} \alpha_n \sigma_{r,H,n} \sigma^2_{t,H_i,m}  \big)$. The upper limit of this bound is obtained via the following relationship
\begin{equation}\label{App-IV-10-1}\nonumber
\begin{split}
\lambda_i  \leq \max_{m} \sum^{M}_{n=1} \frac{\alpha_n \sigma_{r,H,n} \sigma^2_{t,H_i,m}}{( 1+\alpha_n \sigma_{t,H_i,m} \sigma^2_{s_i,m})^2}
%& \leq \max_{m} \sum^{M}_{n=1} \frac{\alpha_n \sigma_{r,H,n} \sigma^2_{t,H_i,m}}{( 1+\alpha_n \sigma_{t,H_i,m} \sigma^2_{s_i,m})^2} =
\leq  \max_{m}  \sum^{M}_{n=1} \alpha_n \sigma_{r,H,n} \sigma^2_{t,H_i,m}.
\end{split}
\end{equation}
%\end{remark}

\vspace{+5pt}
\section{Training Sequence Design for BC Phase}
In this section, we intend to optimize the training sequence ${\bf S}_R$ by minimizing the total estimation MSE at the two source ends subject to the relay power constraint.
According to the MSE derived in \eqref{III-74}, the corresponding optimization for this problem can be formulated as
\begin{equation}\label{IV-31}
\begin{split}
\min_{{\bf S}_R} & ~~\sum^2_{i=1} {\rm Tr}\left[ {\bf C}_{0,G_i} \left({\bf I} + \left({\bf S}^T_R {\bf C}_{t,G}\otimes {\bf C}_{r,G_i} \right)^H {\bf K}^{-1}_i
 \left({\bf S}^T_R {\bf C}_{t,G}\otimes {\bf C}_{r,G_i} \right)  \right )^{-1} \right] \\
{\rm s.t.}& ~~ {\rm Tr}({\bf S}_R {\bf S}^H_R) \leq \tau_R.
\end{split}
\end{equation}
It is also worth noting that the training designs for point-to-point systems in \cite{Biguesh_tsp2009, Bjornson2010} are not applicable to the scenario under consideration here, since the training sequence at the relay, ${\bf S}_R$, needs to be optimized to enhance channel estimation over both links that connect the relay to the sources nodes. Prior to solving \eqref{IV-31}, let us first present the minimum training sequence length required for channel estimation in the BC phase.
\vspace{-7pt}
%\textbf{Corollary} 1:
\begin{lemma}\label{BC_length}
To achieve arbitrary small MSE with sufficiently large relay power, the minimum length of the relay training sequence must be set to $L_R = M$. Otherwise, even with infinite power at the relay, the total MSE is always lower bounded by $\sum^{M}_{m=L_R+1}\sigma_{t,G,m}(\sum^{N_1}_{n=1}\sigma_{r,G_1,n} + \sum^{N_2}_{n=1}\sigma_{r,G_2,n})$, with $\{ \sigma_{t,G,m}\}$ being the eigenvalues of ${\bf Z}_{t,G}$ arranged in decreasing order. Moreover, when ${\bf K}_{q,i} = q_i {\bf I}$, for $i=1,2$, and with any relay power constraint, if the optimal solution of ${\bf S}_R$ in \eqref{IV-31} has a rank of $r$, then the minimum length of the relay training sequence should be $L_R = r$.
\end{lemma}
\vspace{-7pt}
\begin{proof}
Since the proof is similar to the proof of \emph{Lemma~\ref{MAC_length}}, it is omitted for brevity.
%\textcolor{red}{(!Here the proof is very similar to Appendix-A, by using the inequality of eigenvalue.}
\end{proof}

In general, the optimization in \eqref{IV-31} is non-convex. Hence, as in \emph{Algorithm 1}, we first propose an iterative approach to optimize the design of the training sequences at the relay. To this end, the MSE, i.e., $e_i$, given in \eqref{III-74} is rewritten as
\begin{equation}\label{IV-32}
\begin{split}
\tilde{e}_i = & {\rm Tr}\left[{\bf C}_{0,G_i} ({\bf w}_{G_i} - {\bf T}_i {\bf y}_i)({\bf w}_{G_i} - {\bf T}_i {\bf y}_i)^H \right] \\
 = & {\rm Tr}\left[{\bf C}_{0,G_i} - \left({\bf S}^T_i {\bf C}_{t,G}\otimes {\bf C}_{r,G_i} \right)^H {\bf T}^H_i {\bf C}^H_{0,G_i}
 -{\bf C}_{0,G_i} {\bf T}_i \left({\bf S}^T_i {\bf C}_{t,G}\otimes {\bf C}_{r,G_i} \right) \right.\\
& \left. + {\bf C}_{0,G_i} {\bf T}_i \left({\bf S}^T_i {\bf C}_{t,G}\otimes {\bf C}_{r,G_i} \right)
 \left({\bf S}^T_i {\bf C}_{t,G}\otimes {\bf C}_{r,G_i} \right)^H {\bf T}^H_i
 + {\bf C}_{0,G_i} {\bf T}_i {\bf K}_i {\bf T}^H_i \right].
\end{split}
\end{equation}
Since \eqref{III-72} and \eqref{III-74} are in equivalent form,
the optimization problem in \eqref{IV-31} can be rewritten as
\begin{equation}\label{IV-33}
\begin{split}
  \min_{{\bf T}_1, {\bf T}_2, {\bf S}_R}~~&   \tilde{e}_1+  \tilde{e}_2 \\
 {\rm s.t.} ~~~~& {\rm Tr}({\bf S}_R {\bf S}^H_R) \leq \tau_R.
\end{split}
\end{equation}

In the first subproblem, for a given ${\bf S}_R$, the optimal LMMSE estimators ${\bf T}_1$ and ${\bf T}_2$ at the two source ends are obtained as given in \eqref{III-722}. Thus, we focus on solving the second subproblem, where the relay training sequence is optimized for a fixed LMMSE estimator. Similar to \eqref{IV-3} and \eqref{IV-7}, the MSE in \eqref{IV-32} is reexpressed as
\begin{equation}\label{IV-34}\nonumber
\begin{split}
\tilde{e}_i  = & {\bf s}^H_R {\bf E}^H_i \left( {\bf T}^H_i {\bf C}_{0,G_i} {\bf T}_i \otimes {\bf C}^T_{tr,G_i}   \right) {\bf E}_i {\bf s}_R
 - {\rm vec}({\bf C}_{T,G_i})^T {\bf E}_i {\bf s}_R - {\rm vec}({\bf C}^*_{T,G_i})^T {\bf E}_i {\bf s}^*_R \\
&  + {\rm Tr}({\bf C}_{0,G_i}{\bf T}_i {\bf K}_i {\bf T}^H_i)+ {\rm Tr}({\bf C}_{0,G_i}),
\end{split}
\end{equation}
where ${\bf s}_R \triangleq {\rm vec}({\bf S}_R)$, ${\bf C}_{tr,G_i} \triangleq {\bf C}_{t,G} {\bf C}^H_{t,G} \otimes {\bf C}_{r,G_i} {\bf C}^H_{r,G_i}$, and ${\bf C}_{T,G_i}\triangleq({\bf C}_{t,G}\otimes {\bf C}_{r,G_i}) {\bf C}_{0,G_i} {\bf T}_i$. In this case, ${\bf E}_i$ is an $N_iL_R\times ML_R$ matrix that is constructed as in \eqref{IV-3}. Accordingly, \eqref{IV-33} can be rewritten as
\begin{equation}\label{IV-35}
\begin{split}
 \min_{ {\bf S}_R}~~ & {\bf s}^H_R  {\bf A}_R  {\bf s}_R - {\bf a}^T_R {\bf s}_R -  {\bf a}^H_R {\bf s}^*_R   \\
{\rm s.t.} ~~~& {\bf s}^H_R{\bf s}_R \leq \tau_R,
\end{split}
\end{equation}
where ${\bf A}_R \triangleq \sum^2_{i=1}{\bf E}^H_i({\bf T}^H_i {\bf C}_{0,G_i} {\bf T}_i \otimes {\bf C}^T_{tr,G_i}){\bf E}_i $ and ${\bf a}_R \triangleq  {\bf E}^T_1{\rm vec}({\bf C}_{T,G_1})+ {\bf E}^T_2{\rm vec}({\bf C}_{T,G_2})$. Note that different from the source training design, \eqref{IV-35} has only one power constraint. Thus, its solution can be obtained via the KKT conditions given by
\begin{equation}\label{IV-36}
\begin{split}
{\bf A}_R {\bf s}_R - {\bf a}^*_R + \lambda {\bf s}_R  = {\bf 0},~
%\end{split}
%\end{equation}
%\begin{equation}\label{IV-37}
%\begin{split}
\lambda ({\bf s}^H_R{\bf s}_R - \tau_R) = 0,~
%\end{split}
%\end{equation}
%\begin{equation}\label{IV-38}
%\begin{split}
{\bf s}^H_R{\bf s}_R - \tau_R   \leq 0,
\end{split}
\end{equation}
where $\lambda$ is the lagrangian multiplier associated with the relay power constraint.
The closed-form solution of \eqref{IV-35} is obtained as
\begin{equation}\label{IV-39}
\begin{split}
{\bf s}_R = \left({\bf A}_R + \lambda {\bf I} \right)^{-1}  {\bf a}^*_R.
\end{split}
\end{equation}
In \eqref{IV-39}, if the solution ${\bf s}_R$ with $\lambda = 0$ violates the KKT conditions given in \eqref{IV-36}, $\lambda$ should be chosen to meet ${\bf s}^H_R{\bf s}_R = \tau_R$. Consequently, we introduce the following lemma.
%\textcolor{red}{(!This is directly obtained from KKT conditions given in \eqref{IV-36}).}
\vspace{-7pt}
%\textbf{Lemma 2}:
\begin{lemma}\label{BC_KKT}
The function $g(\lambda)={\bf s}^H_R {\bf s}_R$, with ${\bf s}_R$ defined in \eqref{IV-39}, is monotonically decreasing with respect to $\lambda$. Moreover, the optimal $\lambda$ is upper-bounded by ${\displaystyle\sqrt{\frac{\sigma_{a,R}}{\tau_R}}-\sigma_{R,\text{min}}}$ with $\sigma_{a,R} = {\bf a}^T_R {\bf a}^*_R$ and $\sigma_{R,\text{min}}$ denoting the smallest eigenvalue of $ {\bf A}_R $.
\end{lemma}
\vspace{-7pt}
\begin{proof}
Since the proof is similar to that of \emph{Lemma~\ref{MAC_sec_A}}, it is omitted for brevity.
\end{proof}

\noindent Using the above steps, the overall relay training design algorithm can be summarized as follows:

\vspace{-10pt}
\hrulefill
\par
{\footnotesize
\textbf{Algorithm 2}
\begin{itemize}
\item \textbf{Initialize} ${\bf S}_R$
\item \textbf{Repeat}
\begin{itemize}
\item Update the LMMSE estimator matrix ${\bf T}_i$, for $i=1,2$, using \eqref{III-722} for a fixed ${\bf S}_R$;
\item Update the training signal ${\bf S}_R$ using \eqref{IV-39} for a fixed ${\bf T}_i$, for $i=1,2$;
\end{itemize}
\item \textbf{Until} The difference between the MSE from one iteration to another is smaller than a certain predetermined threshold.
\end{itemize}}
\vspace{-10pt}
\hrulefill

Although for the general case the solution of ${\bf S}_R$ can only be obtained via an iterative approach, it is shown that for some special cases, the optimal training sequence of ${\bf S}_R$ can be found in closed-form.

\subsection{When ${\bf K}_{q,i} = q_i {\bf I}$}

In this subsection, we consider that the temporal covariance matrix, ${\bf K}_{q,i}$, is a scalar multiple of the identity matrix. This scenario corresponds to the practical case, where the disturbance consists of both the additive white Gaussian noise and the temporally uncorrelated interference. Using similar steps as in Appendix~\ref{prof_theorem1}, we rewrite the MSE in \eqref{III-74} as
\begin{equation}\label{IV-40} \nonumber
\begin{split}
e_i  =   \sum^{N_i}_{n=1} \sigma_{r,G_i,n} {\rm Tr}\big[ \big({\bf Z}^{-1}_{t,G} + \beta_{i,n}   {\bf S}^*_R {\bf K}^{-1}_{q,i} {\bf S}^T_R  \big )^{-1} \big],
\end{split}
\end{equation}
where $\beta_{i,n} = \frac{\sigma_{r,G_i,n}}{\delta_{r,i,n}}$.

We first consider the scenario where either ${\bf K}_{q,1}$ or ${\bf K}_{q,2}$ is equal to a scalar multiple of the identity matrix. Without loss of generality, let us assume that ${\bf K}_{q,1} = q_1 {\bf I}$, while ${\bf K}_{q,2}$ is assumed to be an arbitrary matrix. Subsequently, we have that
\begin{equation}\label{IV-41} \nonumber
\begin{split}
e_1  =   \sum^{N_1}_{n=1} \sigma_{r,G_1,n} {\rm Tr}\big[ \big({\bf Z}^{-1}_{t,G} + \tilde{\beta}_{1,n}   {\bf S}^*_R {\bf S}^T_R  \big )^{-1} \big],
\end{split}
\end{equation}
where $\tilde{\beta}_{1,n} =  \beta_{1,n}/q_1$. For this case, the original problem in \eqref{IV-31} can be transformed to
\begin{equation}\label{IV-42}
\begin{split}
\min_{{\bf S}_R} & ~~ \sum^{N_1}_{n=1} \sigma_{r,G_1,n} {\rm Tr}\left[ \left({\bf Z}^{-1}_{t,G} + \tilde{\beta}_{1,n}   {\bf S}^*_R {\bf S}^T_R  \right )^{-1} \right] +
 \sum^{N_2}_{n=1}\sigma_{r,G_2,n} {\rm Tr}\left[ \left({\bf Z}^{-1}_{t,G} + \beta_{2,n}   {\bf S}^*_R {\bf K}^{-1}_{q,2} {\bf S}^T_R  \right )^{-1} \right] \\
{\rm s.t.}& ~~~ {\rm Tr}({\bf S}_R {\bf S}^H_R) \leq \tau_R.
\end{split}
\end{equation}
To solve \eqref{IV-42}, the following lemma is introduced.%, and the proof can be found in Appendix~\ref{prof_lemma3}.
\vspace{-7pt}
\begin{lemma}\label{BC_sec_A}
With the minimum relay training sequence length given in Lemma~\ref{BC_length}, i.e., $L_R = M$, when the eigenvalues of ${\bf Z}_{t,G}$, ${\bf Z}_{r,G_i}$, ${\bf K}_{q,i}$, and ${\bf K}_{r,i}$ are small compare to one,
%\textcolor{red}{(!Here, if we denote the one eigenvalue by $x$, i mean $x$ is small if $x$ is samller than $1/x$. In this case, the original problem can be simplified. In practise, I think it implies that both channel gains and the noise variance are very small.)},
the optimal ${\bf S}_R$ in \eqref{IV-42} has the following structure
%in the high signal-to-noise ratio (SNR) regime
\begin{equation}\label{IV-43}
%\begin{split}
{\bf S}_R = {\bf U}^*_{t, G} {\bm \Sigma}_{s,R} {\bf U}^T_{q,2},
%\end{split}
\end{equation}
where ${\bm \Sigma}_{s,R}$ is the diagonal eigenvalue matrix.
{In \eqref{IV-43}, taking ${\bf U}_{t, G}$ and ${\bf U}^T_{q,2}$ as the eigenvector matrices of ${\bf Z}_{t,G}$ and ${\bf K}_{q,2}$, respectively,
the eigenvalues of ${\bf Z}_{t,G}$ and ${\bf K}_{q,2}$ should be arranged in an opposite order.}
With the structure given in \eqref{IV-43} and by defining $[{\bm \Sigma}_{s,R}]_{m,m} \triangleq \sigma^{1/2}_{s,R,m}$, $\sigma_{s,R,m}$ can be obtained from
\begin{equation}\label{IV-44}
%\begin{split}
\lambda = \sum^{N_1}_{n=1}  \frac{\sigma_{r,G_1,n} \tilde{\beta}_{1,n}\sigma^2_{t,G,m}  }{(1+\tilde{\beta}_{1,n}\sigma_{t,G,m} \sigma_{s,R,m})^2 } +
  \sum^{N_2}_{n=1}  \frac{\sigma_{r,G_2,n} {\beta}_{2,n}  \delta^{-1}_{q,2,m}\sigma^2_{t,G,m}}{(1+{\beta}_{2,n} \sigma_{t,G,m} \sigma_{s,R,m} \delta^{-1}_{q,2,m})^2 },
%\end{split}
\end{equation}
where $\lambda\in\big(0, \max_{m}\big\{ \sum^{N_1}_{n=1} ( \sigma_{r,G_1,n}\tilde{\beta}_{1,n}\sigma^2_{t,G,m} + \frac{ \sigma_{r,G_2,n} {\beta}_{2,n}   \sigma^2_{t,G,m}}{\delta_{q,2,m}}  ) \big\}  \big)$ and can be found via the bisection search algorithm to meet the relay's power constraint.
\end{lemma}
\vspace{-7pt}
\begin{proof}
Please refer to Appendix~\ref{prof_lemma3}.
\end{proof}

For the case with ${\bf K}_{q,1} = q_1 {\bf I}$ and ${\bf K}_{q,2} = q_2 {\bf I}$, \eqref{IV-42} can be further simplified as
\begin{equation}\label{IV-45}
\begin{split}
\min_{{\bf S}_R} & ~~ \sum^{N_1}_{n=1} \sigma_{r,G_1,n} {\rm Tr}\left[ \left({\bf Z}^{-1}_{t,G} + \tilde{\beta}_{1,n}   {\bf S}^*_R {\bf S}^T_R  \right )^{-1} \right] +
 \sum^{N_2}_{n=1}\sigma_{r,G_2,n} {\rm Tr}\left[ \left({\bf Z}^{-1}_{t,G} + \tilde{\beta}_{2,n}    {\bf S}^*_R  {\bf S}^T_R  \right )^{-1} \right] \\
{\rm s.t.}& ~~~ {\rm Tr}({\bf S}_R {\bf S}^H_R) \leq \tau_R,
\end{split}
\end{equation}
where $ \tilde{\beta}_{2,n}=  \beta_{2,n}/q_2$. The optimal solution of \eqref{IV-45} is given in the following lemma.

\begin{lemma}\label{BC_sec_AA}
The optimal ${\bf S}_R$ in \eqref{IV-45} is of the form
\begin{equation}\label{IV-46}\nonumber
\begin{split}
{\bf S}_R = {\bf U}^*_{t, G} {\bm \Sigma}_{s,R} {\bf V}^T_{s,R},
\end{split}
\end{equation}
where ${\bm \Sigma}_{s,R}$ is a positive real diagonal matrix, ${\bf U}_{t, G}$ is the eigenvector matrix of ${\bf Z}_{t,G}$, and ${\bf V}_{s,R}$ is an arbitrary unitary matrix. The corresponding eigenvalues of ${\bf Z}_{t,G}$ are arranged in the same order as the diagonal elements of ${\bm \Sigma}_{s,R}$. The optimal ${\bm \Sigma}_{s,R}$ can be obtained from
\begin{equation}\label{IV-47}\nonumber
\begin{split}
\lambda =  \sum^{N_1}_{n=1}   \frac{\sigma_{r,G_1,n}  \tilde{\beta}_{1,n} \sigma^2_{t,G,m} }{(1+\tilde{\beta}_{1,n} \sigma_{t,G,m} \sigma_{s,R,m})^2 } +
  \sum^{N_1}_{n=1} \frac{\sigma_{r,G_2,n}  \tilde{{\beta}}_{2,n} \sigma^2_{t,G,m} }{(1+\tilde{{\beta}}_{2,n} \sigma_{t,G,m} \sigma_{s,R,m} )^2 },
\end{split}
\end{equation}
where $\lambda\in \big(0, \max_{m}\big\{ \sum^{N_1}_{n=1} (\sigma_{r,G_1,n}\tilde{\beta}_{1,n}\sigma^2_{t,G,m} +\sigma_{r,G_2,n}\tilde{\beta}_{2,n} \sigma^2_{t,G,m}) \big\}  \big)$ can be found via the bisection search to meet the relay's power constraint.
\end{lemma}
\begin{proof}
The expression of MSE, i.e., $e_i$, for $i=1,2$, in \eqref{IV-45} contains a term ${\rm Tr}\big[ \big({\bf Z}^{-1}_{t,G} + \tilde{\beta}_{i,n}   {\bf S}^*_R {\bf S}^T_R  \big )^{-1} \big]$. If the eigenvalues of ${\bf Z}_{t,G}$ and $ {\bf S}^*_R  {\bf S}^T_R$ are arranged in the same order, we have \cite{Book_Majorization}
\begin{equation}\label{IV-47-1}
{\bm \lambda}({\bf Z}^{-1}_{t,G} )+ {\bm \lambda}(\tilde{\beta}_{i,n}   {\bf S}^*_R {\bf S}^T_R ) \preccurlyeq {\bm \lambda}({\bf Z}^{-1}_{t,G} + \tilde{\beta}_{i,n}   {\bf S}^*_R {\bf S}^T_R ).
\end{equation}
Since the function ${\rm Tr}[({\bf A})]$ is a schur convex function with respect to the eigenvalues of $\bf A$ \cite{Bjornson2010}, based on \eqref{IV-47-1}, we can easily obtain the results in \emph{Lemma~\ref{BC_sec_AA}}.
%
%Since $e_1$ and $e_2$ are dependent on the same matrix ${\bf C}^H_{t,G} {\bf S}^*_R {\bf S}^T_R {\bf C}_{t,G}$, the objective function in \eqref{IV-46} has a simple form than the one in \eqref{IV-42}. Similar to the proof given in Appendix~\ref{prof_lemma3}, by introducing a new variable ${\bf P}$, we can easily obtain \eqref{IV-46}.
\end{proof}

\subsection{When ${\bf Z}_{t,G} = a {\bf I}$}
This case corresponds to a scenario, where the relay antennas are far enough from one another such that they are spatially uncorrelated. In this case, the corresponding training design problem can be formulated as
\begin{equation}\label{IV-48}
\begin{split}
\min_{{\bf S}_R} & ~~ \sum^{N_1}_{n=1} \tilde{\sigma}_{r,G_1,n}  {\rm Tr}\left[ \left( {\bf I} + \bar{{\beta}}_{1,n}   {\bf S}^*_R {\bf K}^{-1}_{q,1} {\bf S}^T_R  \right )^{-1} \right] +
 \sum^{N_2}_{n=1} \tilde{\sigma}_{r,G_2,n}  {\rm Tr}\left[ \left( {\bf I} + \bar{\beta}_{2,n}   {\bf S}^*_R {\bf K}^{-1}_{q,2} {\bf S}^T_R  \right )^{-1} \right] \\
{\rm s.t.}& ~~~ {\rm Tr}({\bf S}^*_R {\bf S}^T_R) \leq \tau_R,
\end{split}
\end{equation}
where $\tilde{\sigma}_{r,G_i,n} = a {\sigma}_{r,G_i,n}$, and $\bar{{\beta}}_{i,n} =  a{\beta}_{i,n}$ for $i=1,2$. In general, the optimization problem in \eqref{IV-48} is non-convex with respect to ${\bf S}_R$. However, based on the minimum training sequence length derived in \emph{Lemma~\ref{BC_length}}, i.e., $L_R = M$, \eqref{IV-48} is equivalent to the following problem
\begin{equation}\label{IV-49}
\begin{split}
\min_{{\bf S}_R} & ~~ \sum^{N_1}_{n=1} \tilde{\sigma}_{r,G_1,n} {\rm Tr}\left[ \left({\bf I} + \bar{{\beta}}_{1,n}  {\bf S}^T_R {\bf S}^*_R {\bf K}^{-1}_{q,1}   \right )^{-1} \right] +
 \sum^{N_2}_{n=1}\tilde{\sigma}_{r,G_2,n} {\rm Tr}\left[ \left({\bf I} + \bar{\beta}_{2,n}   {\bf S}^T_R{\bf S}^*_R {\bf K}^{-1}_{q,2}   \right )^{-1} \right] \\
{\rm s.t.}& ~~~ {\rm Tr}({\bf S}^T_R{\bf S}^*_R ) \leq \tau_R.
\end{split}
\end{equation}
In obtaining \eqref{IV-49} from \eqref{IV-48}, we have used the identity
%\begin{equation}\label{App-I-00}\nonumber
%\begin{split}
$ {\rm Tr}\big( \big[ {\bf I} + {\bf A} {\bf B} \big]^{-1} \big)
 =  {\rm Tr}\big( \big[ {\bf I} + {\bf B} {\bf A}  \big]^{-1} \big)+ m-n$,
%\end{split}
%\end{equation}
where $\bf A$ and $\bf B$ are $m\times n$ and $n\times m$ matrices, respectively, \cite{Cookbook}. Similar to the Section \ref{sec:4-b}, it is observed that in \eqref{IV-49}, we can directly optimize the matrix ${\bf S}^T_R {\bf S}^*_R$ instead of ${\bf S}_R$ by solving the following optimization
\begin{equation}\label{IV-49-1}
\begin{split}
\min_{{\bf Q}_R \succeq {\bf 0}} & ~~ \sum^{N_1}_{n=1}\tilde{\sigma}_{r,G_1,n} {\rm Tr}\left[ \left({\bf I} + \bar{{\beta}}_{1,n}  {\bf K}^{-1/2}_{q,1}{\bf Q}_R {\bf K}^{-1/2}_{q,1}   \right )^{-1} \right] +
 \sum^{N_2}_{n=1} \tilde{\sigma}_{r,G_2,n}{\rm Tr}\left[ \left({\bf I} + \bar{\beta}_{2,n}  {\bf K}^{-1/2}_{q,2} {\bf Q}_R {\bf K}^{-1/2}_{q,2}   \right )^{-1} \right] \\
{\rm s.t.}& ~~~ {\rm Tr}({\bf Q}_R ) \leq \tau_R,
\end{split}
\end{equation}
where ${\bf Q}_R\triangleq {\bf S}^T_R {\bf S}^*_R$.
By using similar steps as that in \emph{Theorem~\ref{MAC_convex}}, it can be shown that \eqref{IV-49-1} is convex and it can be readily solved via the following SDP problem
\begin{equation}\label{IV-50}\nonumber
\begin{split}
 \min_{{\bf Q}_R \succeq {\bf 0},{\bf X}_n,{\bf Y}_n }~~& \sum^{M}_{n=1} \tilde{\sigma}_{r,G_1,n} {\rm Tr}\left({\bf X}_n \right) + \tilde{\sigma}_{r,G_2,n}{\rm Tr}\left({\bf Y}_n \right) \\
 {\rm s.t.} ~~~~~& {\rm Tr}({\bf Q}_R) \leq \tau_R \\
%&  {\bf Q}_R \succeq {\bf 0} \\
&  \left[
     \begin{array}{cc}
       {\bf I} +\bar{{\beta}}_{1,n}  {\bf K}^{-1/2}_{q,1}{\bf Q}_R {\bf K}^{-1/2}_{q,1} & {\bf I} \\
       {\bf I} & {\bf X}_n \\
     \end{array}
   \right]
 \succeq {\bf 0},~\forall n \\
 &  \left[
     \begin{array}{cc}
       {\bf I} +\bar{\beta}_{2,n}  {\bf K}^{-1/2}_{q,2} {\bf Q}_R {\bf K}^{-1/2}_{q,2} & {\bf I} \\
       {\bf I} & {\bf Y}_n \\
     \end{array}
   \right]
 \succeq {\bf 0},~\forall n .
\end{split}
\end{equation}

\section{Simulation results}
In this section, we present simulation results to verify the performance of the proposed training design algorithms. The total normalized MSE (NMSE), defined as either $\frac{1}{ M(N_1+N_2)} \sum^2_{i=1}\mathbb{E}\left\{ || {\bf H}_{i} - \hat{{\bf H}}_{i}||^2_F \right\}$ or \\ $\frac{1}{M(N_1+N_2)}\sum^2_{i=1} \mathbb{E}\left\{ || {\bf G}_{i} - \hat{{\bf G}}_{i} ||^2_F\right\}$, is utilized to illustrate the performance of the proposed algorithms.
In all simulations, the channel covariance matrices are assumed to have the following structures
\begin{equation}\label{V-1}\nonumber
\begin{split}
[{\bf Z}_{t,b}]_{n,m} = z_{t,b} J_0 (d_{t,b}|n-m|),~  b \in \{H_1,H_2,G\}, \notag \\
[{\bf Z}_{r,b}]_{n,m} =  z_{r,b} J_0 (d_{r,b}|n-m|),~  b \in \{H,G_1,G_2\},
\end{split}
\end{equation}
where $J_0(\cdot)$ is the zeroth-order Bessel function of the first kind, $d_{t,b}$ and $d_{r,b}$ are proportional to the carrier frequency and the antenna separation vectors at the transmitter
and the receiver, respectively \cite{Biguesh_tsp2009}. Moreover, the scalars $z_{t,b}$ and $z_{r,b}$ are normalization factors such that ${\rm Tr}({\bf Z}_{t,H_i}) =N_i $, ${\rm Tr}({\bf Z}_{r,H}) =M$, ${\rm Tr}({\bf Z}_{t,G}) =M$ and ${\rm Tr}({\bf Z}_{r,G_i}) =N_i$. The temporal covariance matrix of the disturbance is assumed to be modeled via a first order autoregressive (AR) filter, i.e., ${\rm AR}(1)$, that is denoted by $[{\bf K}_{q,b}]_{n,m}= I_{q,b} k_{q,b} \eta^{|n-m|}_{q,b}$ for $b \in\{1,2,R \}$ \cite{Biguesh_tsp2009}. Here, the scalar $k_{q,b}$ is a normalization factor similar to ${\bf Z}_{t,b}$ and ${\bf Z}_{r,b}$. Moreover, $I_{q,b}$ indicates the strength of the interference from the nearby users. Following the approach in \cite{Bjornson2010}, it is assumed that the received spatial covariance matrix of the disturbance, ${\bf K}_{r,b}$, shares the same eigenvalue vectors with ${\bf Z}_{r,b}$ but with different eigenvalues. For simplicity, the length of the source and relay training sequences are assumed to be $L_S = N_1+N_2$ and $L_R = M$, respectively. The sum power at the two sources are assumed to be $\tau_1+\tau_2 = 2P$.
If not specified otherwise, we assume that $N_1=N_2=M=3$. Furthermore, the system parameters for the MAC phase are set to:
$d_{t,H_1} = 1.5$, $d_{t,H_2} = 1.8$, $d_{r,G} = 1.3$, $\eta_{q,R}=0.9$ and $I_{q,R}=1$, while for the BC phase, we choose $d_{t,G} = 1.9$, $d_{r,G_1} = 1.95$, $d_{r,G_2} = 0.3$, $\eta_{q,1}=0.9$, $\eta_{q,2}=-0.9$ and $I_{q,1}= I_{q,2} =1$.

\begin{figure}[t]
  \centering
  \subfigure[MAC phase]{
    \label{Coverg:subfig:a} %% label for first subfigure
    \includegraphics[scale=0.54]{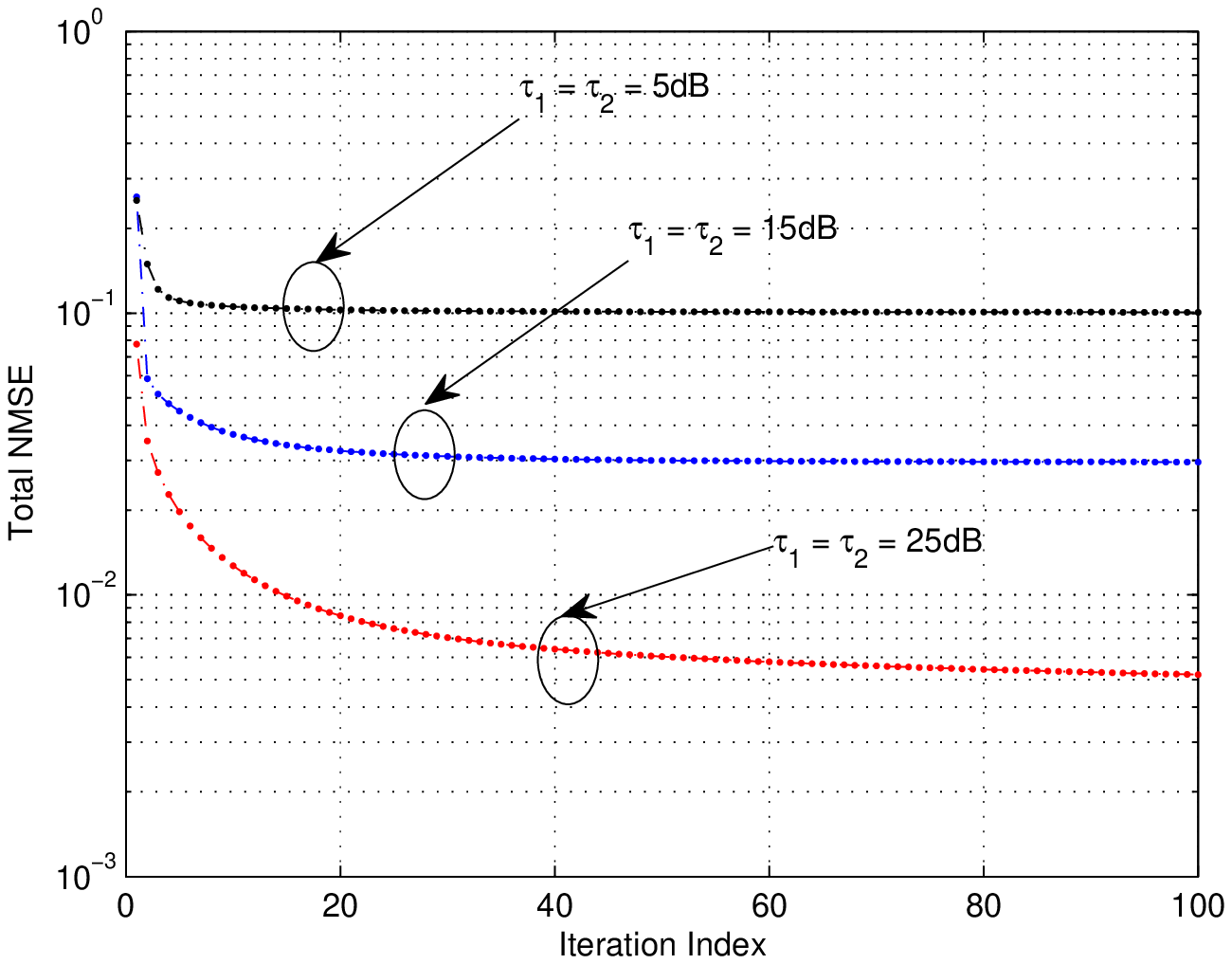}}
  \hspace{0in}
  \subfigure[BC phase]{
    \label{Coverg:subfig:b} %% label for second subfigure
    \includegraphics[scale=0.54]{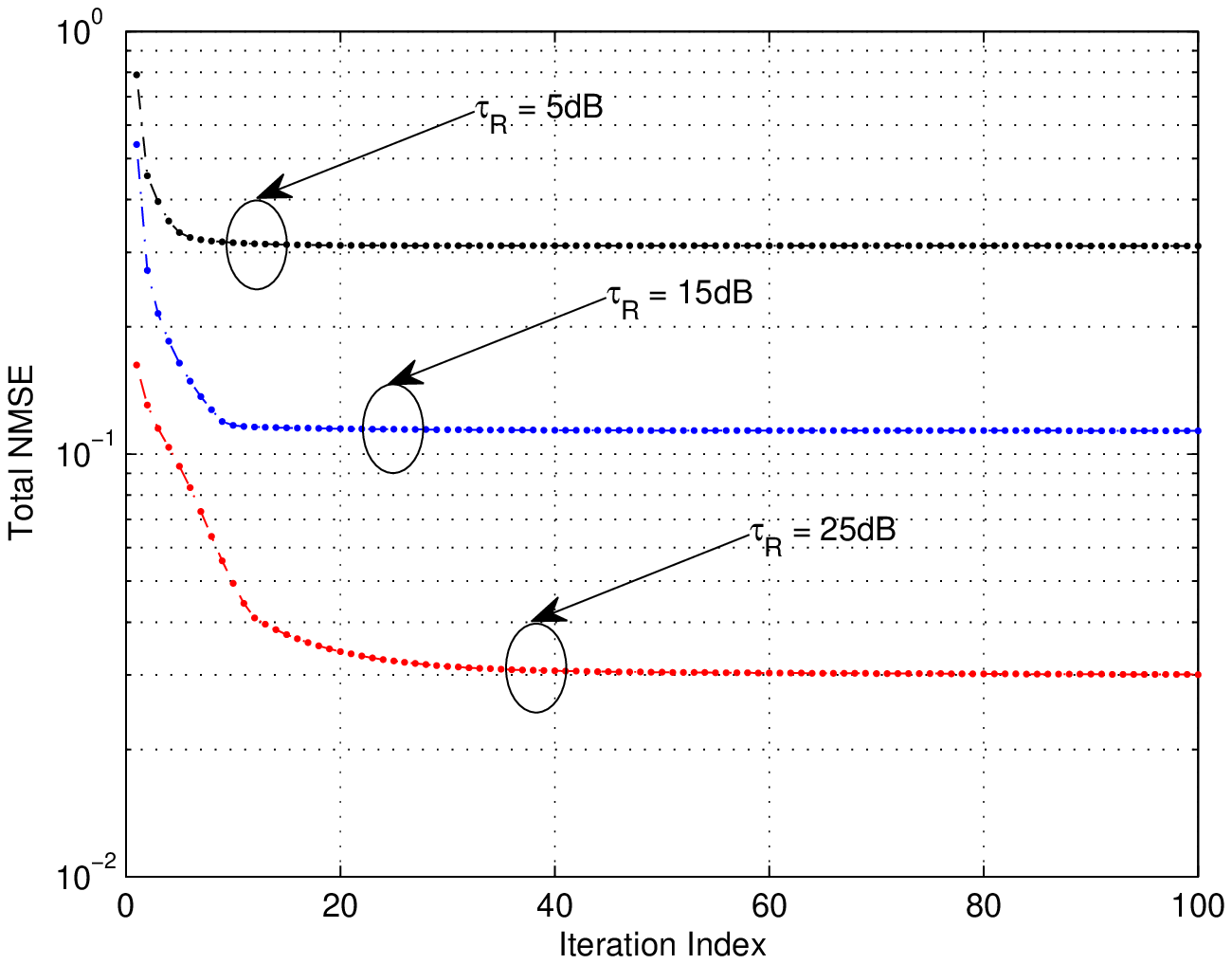}}
  \caption{Convergence behavior of the proposed iterative designs.}
  \label{Coverg:subfig} %% label for entire figure
\end{figure}

In Fig.~\ref{Coverg:subfig}, the convergence behaviors of \emph{Algorithms 1} and \emph{2} for different SNRs are shown in subfigures $(a)$ and $(b)$, respectively. It is illustrated that in general, the proposed iterative algorithms converge very quickly and at most $60$ iterations are required for them to converge. These results also indicate that as the SNR increases more iterations are needed for the proposed algorithms to converge. In Figs.~\ref{Optim:subfig:a} and~\ref{Optim:subfig:b}, the convergence of the proposed algorithms are verified for different sets of initializations. In this setup, ``Random-$1$" indicates that a random initial point is selected, ``Random-$N$" implies that $N$ random initial points are tested but the one with the best performance is selected, and ``Identity" indicates that
\begin{equation}\label{V-1-1}\nonumber
{\bf S}={\rm Blkdiag}(a{\bf I}_{N_1}, b{\bf I}_{N_2})~{\rm and}~{\bf S}_R = c{\bf I}_M,
\end{equation}
where $a$, $b$, and $c$ are used to satisfy the source and relay power constraints.
%the identity matrix scaled by a factor is chosen as an initial point.
The results in Figs.~\ref{Optim:subfig:a} and~\ref{Optim:subfig:b} indicate that for various SNR values, the proposed iterative training design algorithms are not sensitive to the selected initial point. Furthermore, it is observed that the initialization process denoted by ``Identity" performs well as an initial point and can approach the initialization scenario denoted by "Random-$10$". Hence, in the following, if not state otherwise, the ``Identity" initialization point is used.
\begin{figure}[t]
  \centering
  \subfigure[MAC phase]{
    \label{Optim:subfig:a} %% label for first subfigure
    \includegraphics[scale=0.54]{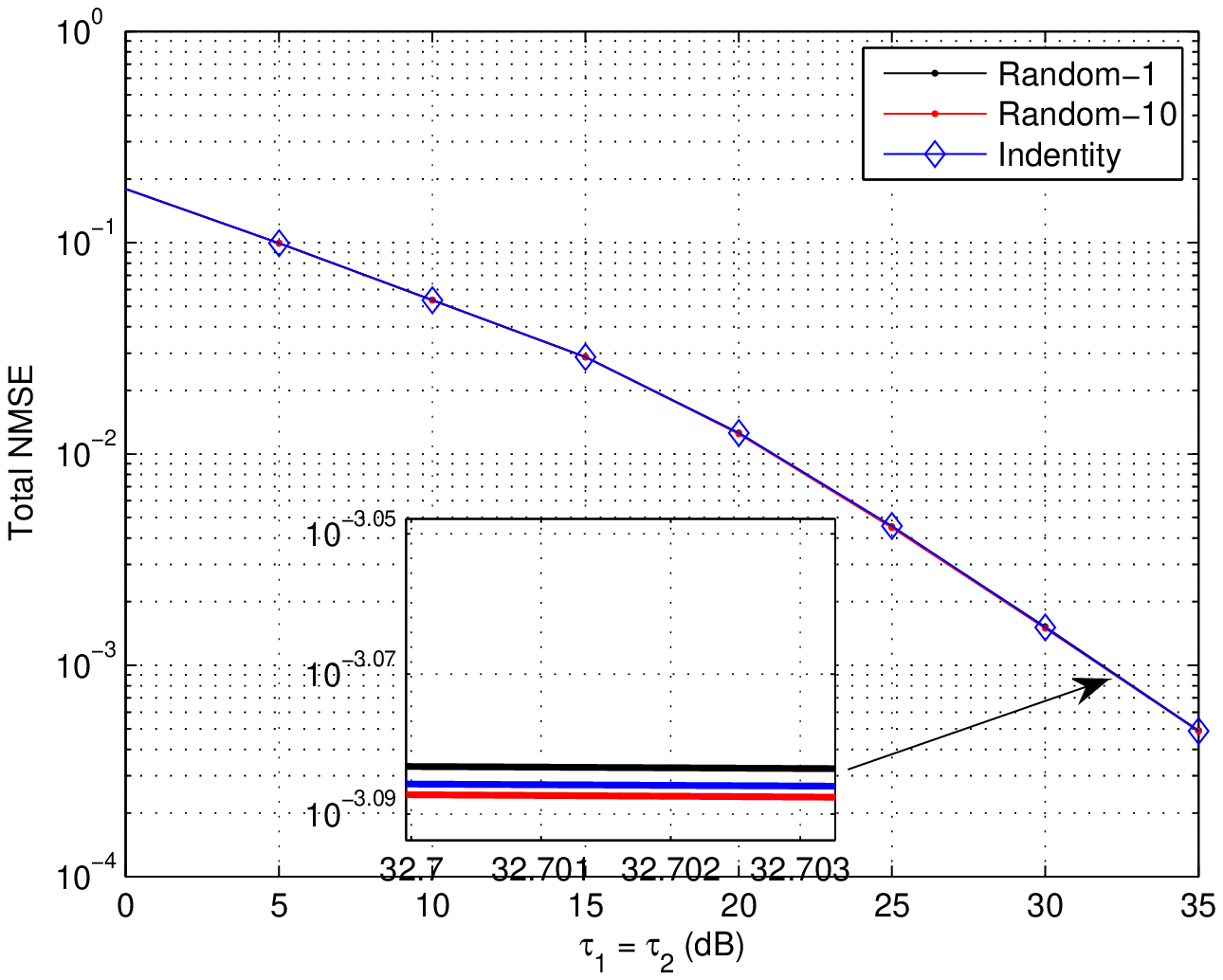}}
  \hspace{0in}
  \subfigure[BC phase]{
    \label{Optim:subfig:b} %% label for second subfigure
    \includegraphics[scale=0.54]{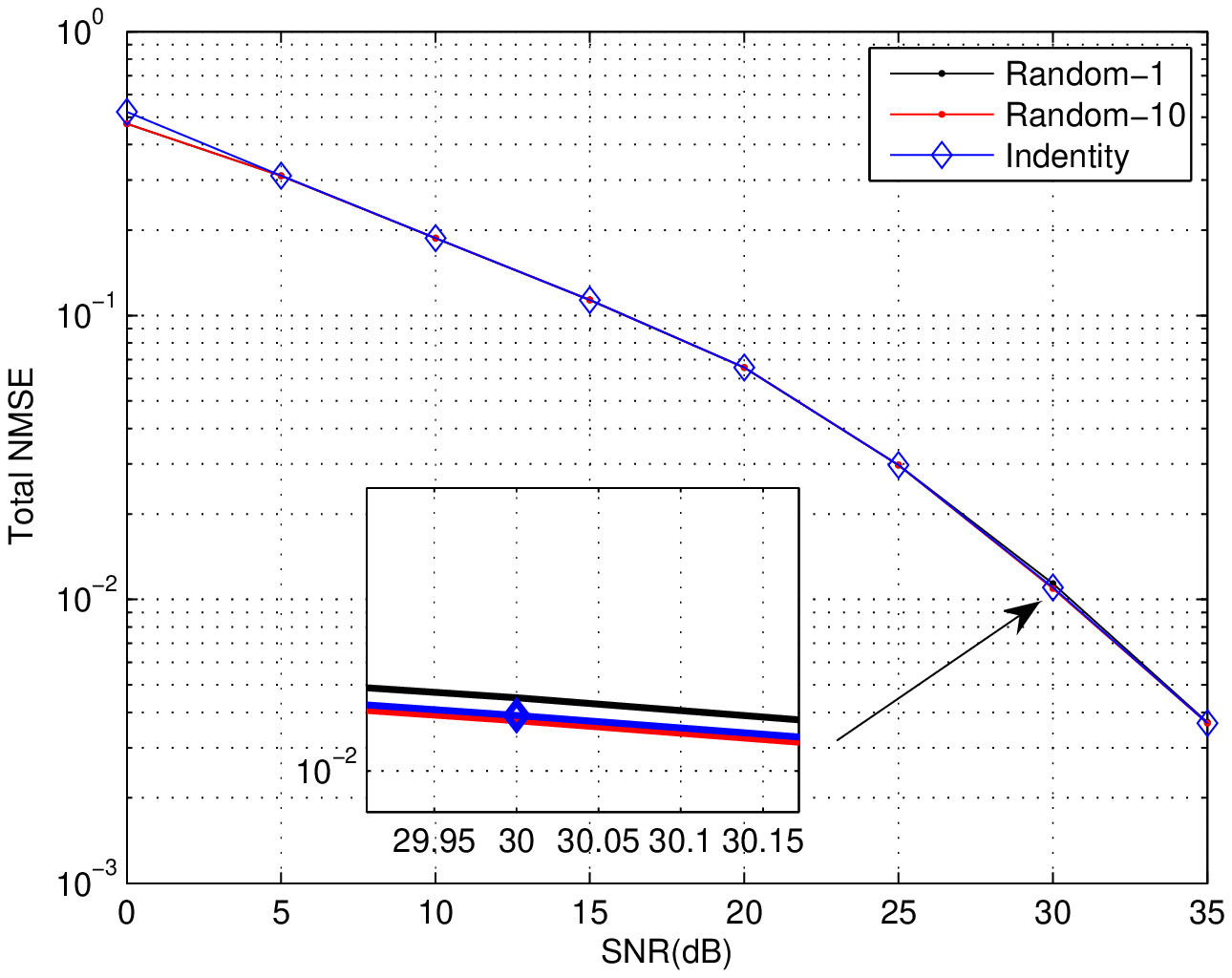}}
  \caption{Optimality for the proposed iterative designs with different initiations.}
  \label{Optim:subfig} %% label for entire figure
\end{figure}

In Fig.~\ref{Comp:subfig}, we compare the total NMSE of the proposed iterative training design algorithm with that of \cite{Bjornson2010}, which is intended for point-to-point systems. To make this comparison possible, for the training sequence design in the MAC phase, it is assumed that two source nodes transmit their training sequences in two orthogonal time intervals, i.e., ${\bf s}_1(t)$ with $t \in [1,2,\cdots,N_1]$ and ${\bf s}_2(t)$ with $t \in [N_1+1,N_1+2,\cdots,N_1+N_2]$. In the BC phase, the training sequence, ${\bf S}_R$ is designed according to the channel from the relay to the source $S_1$. The plots in Figs.~\ref{Comp:subfig:a} and~\ref{Comp:subfig:b} illustrate that compared to the approach in~\cite{Bjornson2010}, the proposed training design can significantly improve the accuracy of channel estimation in TWR systems. This gain is even more pronounced when the two source nodes operate at different transmit power levels during the MAC phase and when the strengths of the disturbances at the two source nodes are asymmetric, i.e., $I_{q,1} \neq I_{q,2}$ during the BC phase.
This can be mainly attributed to the fact that the proposed training design algorithm, i.e., \emph{Algorithm 1}, takes into account the temporal correlation of the disturbances at the relay node in the MAC phase, while ensuring that the training sequences transmitted from the relay node simultaneously match the channels corresponding to relay-to-source links during the BC phase.

\begin{figure}[t]
  \centering
  \subfigure[MAC phase]{
    \label{Comp:subfig:a} %% label for first subfigure
    \includegraphics[scale=0.54]{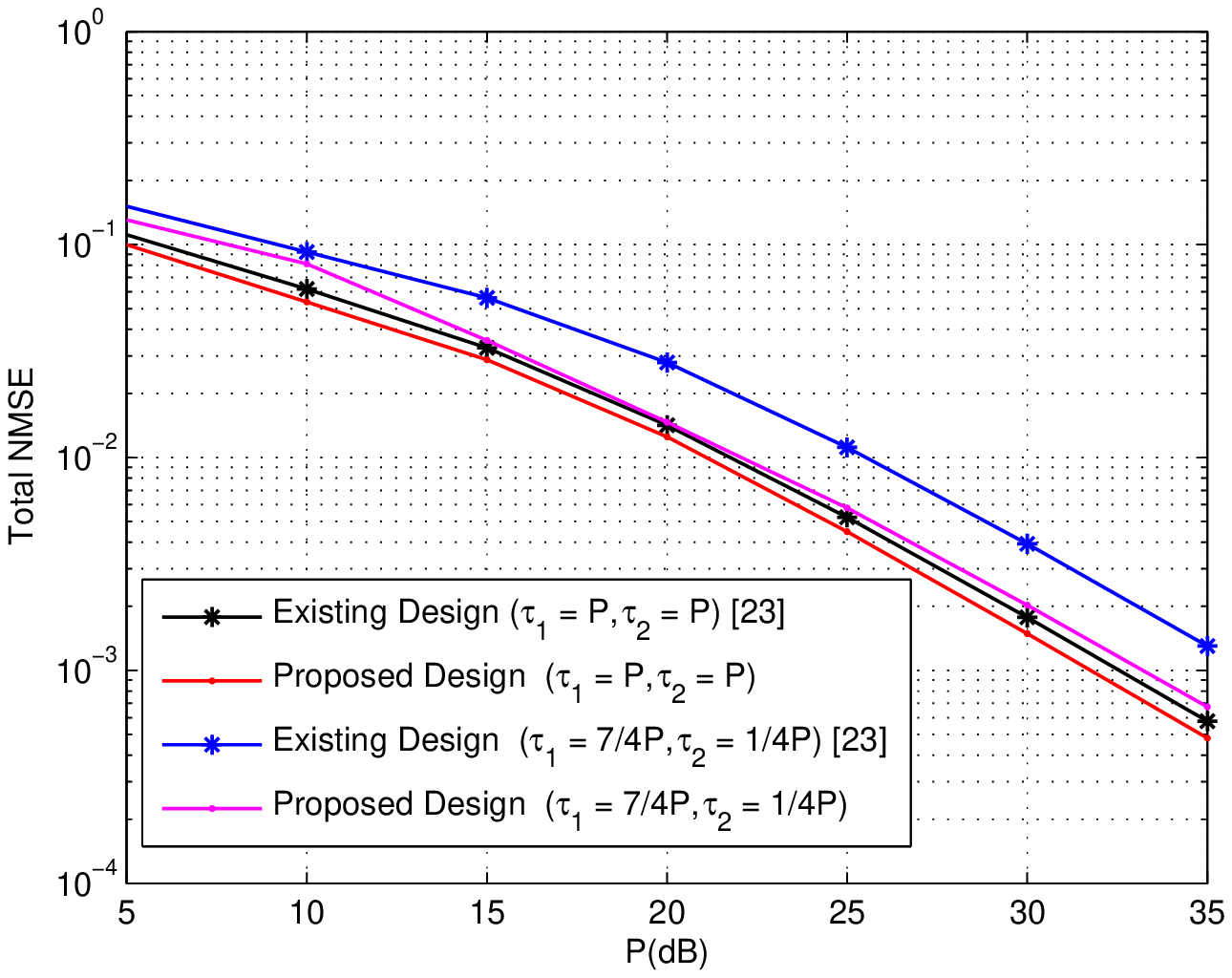}}
  \hspace{0in}
  \subfigure[BC phase]{
    \label{Comp:subfig:b} %% label for second subfigure
    \includegraphics[scale=0.54]{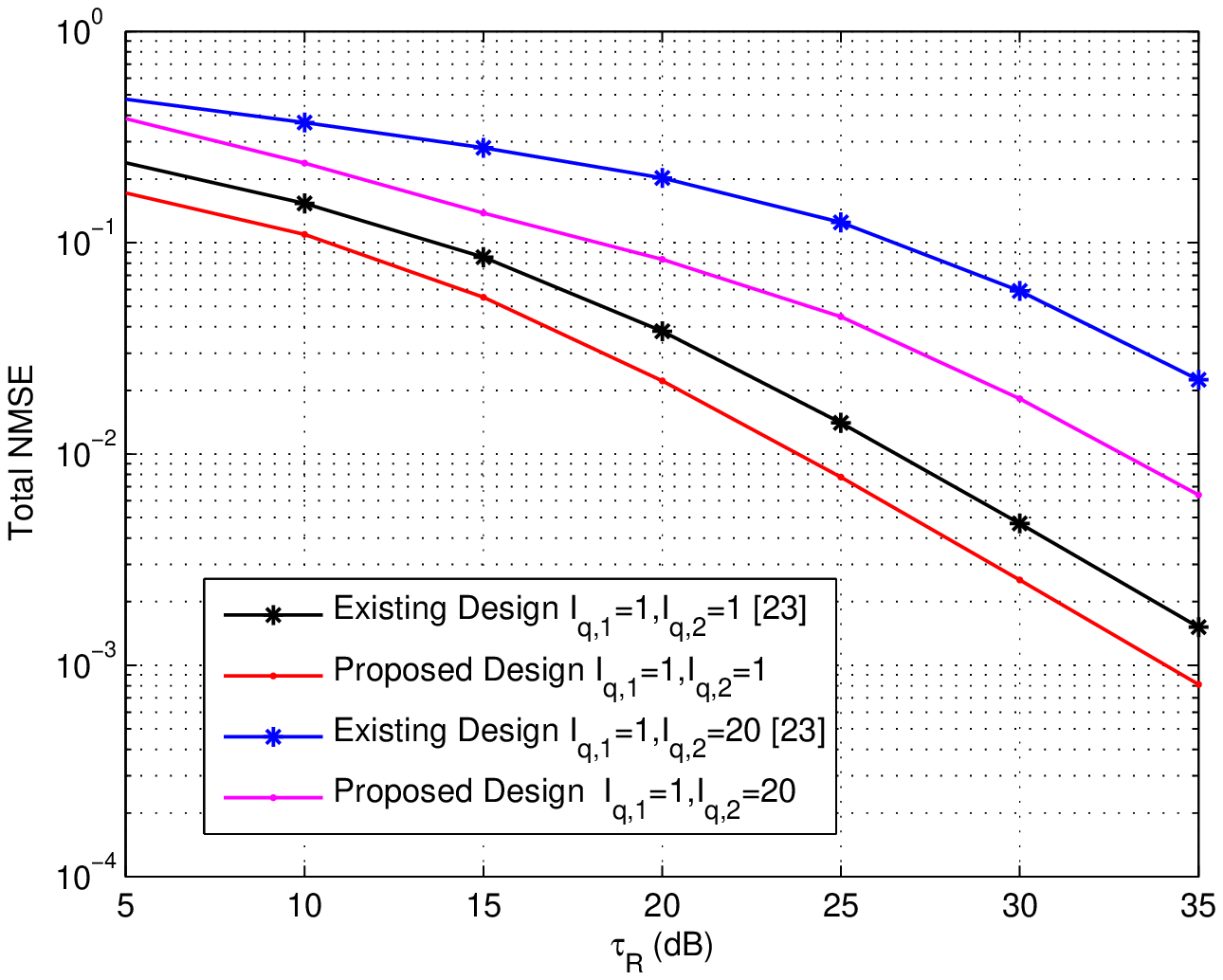}}
  \caption{Performance comparison with the existing design.}
  \label{Comp:subfig} %% label for entire figure
\end{figure}

\begin{figure}[t]
\begin{centering}
\includegraphics[scale=0.6]{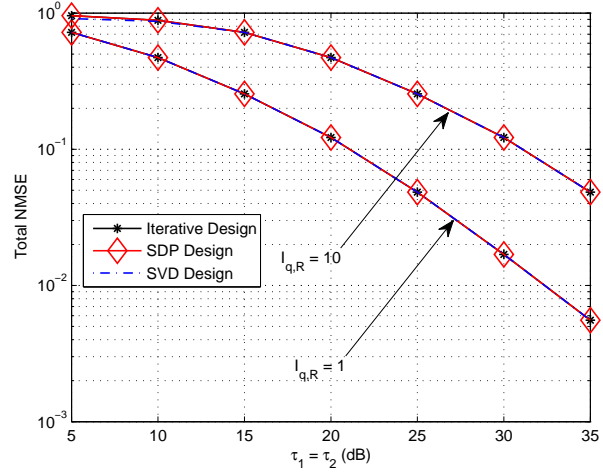}
\vspace{-0.1cm}
\caption{Performance illustration of the case with ${\bf K}_{q,R} = q {\bf I}$ for MAC phase channel estimation.}\label{MAC_Case_II}
\end{centering}
\vspace{-0.3cm}
\end{figure}

In Fig.~\ref{MAC_Case_II}, the performance of the proposed training sequence design algorithms and channel estimators in the MAC phase for ${\bf K}_{q,R} = q {\bf I}$ is demonstrated. Three training sequence design approaches are taken into consideration: 1) The iterative design based on \emph{Algorithm 1}; 2) The SDP design based on \eqref{IV-30}; and 3) The SVD design based on \emph{Theorem~\ref{MAC_sec_B}}. As shown in \emph{Theorem~\ref{MAC_convex}}, in this case, the optimization problem for finding the optimal training sequences is convex. Hence, it is well-known that both the SDP and the SVD design schemes can achieve optimal channel estimation performance. This outcome is also verified by the results in Fig.~\ref{MAC_Case_II}. However, it is interesting to note that the proposed iterative algorithm denoted by \emph{Algorithm 1} can also achieve optimal performance, which further verifies its effectiveness for designing the training sequences in the MAC phase.

\begin{figure}[t]
\begin{centering}
\includegraphics[scale=0.6]{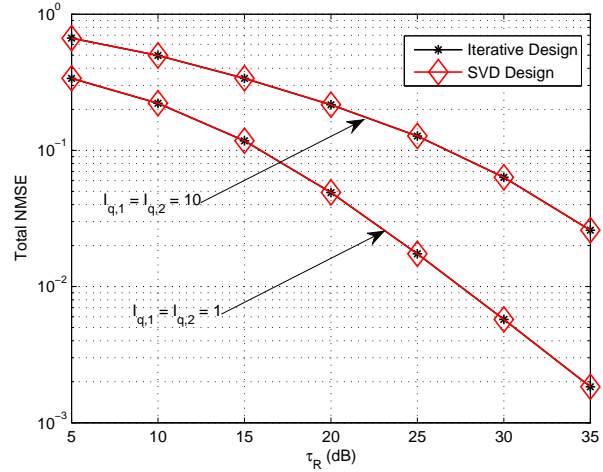}
\vspace{-0.1cm}
\caption{Total NMSE of channel estimation during the BC phase with ${\bf K}_{q,1} = q_1 {\bf I}$.}\label{BC_Case_I}
\end{centering}
\vspace{-0.3cm}
\end{figure}

%\begin{figure}[t]
%\begin{centering}
%\includegraphics[scale=0.55]{figures/BC_Case_II.eps}
%\vspace{-0.1cm}
%\caption{Performance illustration of case with ${\bf K}_{q,1} = q_1 {\bf I}$ and ${\bf K}_{q,2} = q_2 {\bf I}$ for BC phase channel estimation}\label{BC_Case_II}
%\end{centering}
%\vspace{-0.3cm}
%\end{figure}

\begin{figure}[t]
\begin{centering}
\includegraphics[scale=0.6]{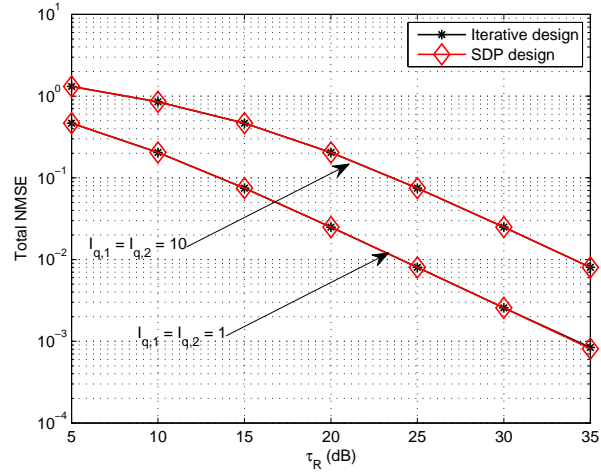}
\vspace{-0.1cm}
\caption{Total NMSE of channel estimation during the BC phase with ${\bf C}_{t,G} = \sqrt{a} {\bf I}$.}\label{BC_Case_III}
\end{centering}
\vspace{-0.3cm}
\end{figure}

Figs.~\ref{BC_Case_I} and~\ref{BC_Case_III} present the channel estimation performance for the special cases presented in Sections V-A and V-B during the BC phase. More specifically, in Fig.~\ref{BC_Case_I}, the scenario where ${\bf K}_{q,1}=q_1 {\bf I}$ and ${\bf K}_{q,2}$ is an arbitrary matrix is taken into consideration. In this figure, the plot with the legend ``SVD design" refers to the results in \emph{Lemma~\ref{BC_KKT}}. Although \emph{Lemma~\ref{BC_KKT}} shows that the training design structure given in \eqref{IV-43} is only optimal when the eigenvalues of ${\bf Z}_{t,G}$, ${\bf Z}_{r,G_i}$, ${\bf K}_{q,i}$ and ${\bf K}_{r,i}$ are small, the results in Fig.~\ref{BC_Case_I} show that the proposed ``SVD design" method closely matches the performance of the proposed iterative design algorithm for most practical scenarios of interest. Moreover, the optimality of \emph{Algorithm 2} is further verified via the results in Fig.~\ref{BC_Case_III}. In this setup we consider the special case presented in Section V-B. Note that since the proposed ``SDP design" in Section V-B is optimal, we conclude that the iterative design in Fig.~\ref{BC_Case_III} can approach the optimal solution in this special case.
%\textcolor{red}{(! Here I just want say that even the eigenvalues are not very small as we simulated, the obtained results shown in Lemma 3 are also almost optimal.)}

\section{Conclusions}
In this paper, the problem of channel estimation in MIMO TWR systems was analyzed. Unlike prior work in this field, the impact of the interference from neighboring devices and the effect of antenna correlations on the design of training sequences and channel estimation performance were taken into consideration. To obtain the channel parameters corresponding to each individual link, we have proposed to carry out the channel estimation process in two phases: the BC phase and the MAC phase. Next, the optimal LMMSE channel estimators for both phases were derived and the corresponding training sequence design problems for each phase were formulated. Subsequently, to ensure accurate channel estimation in TWR systems, the minimum required length of the training sequences were also analyzed. Since the resulting optimization problems were non-convex in their general form, specific transformations were used to obtain near optimal iterative algorithms for the design of the training sequences. Further analysis showed that the optimal structures of the training sequences can be obtained in closed-form when the channel or the noise temporal covariance matrices have special structures. Simulation results show that the proposed training sequence design algorithms can significantly enhance channel estimation performance in TWR relaying systems

%=================================Appendix============================
\appendices
\numberwithin{equation}{section}
\section{Proof of Lemma~\ref{MAC_length}}
\label{prof_theorem1}
To prove \emph{Lemma~\ref{MAC_length}}, we first rewrite \eqref{III-7} into the following form
\begin{equation}\label{App-I-0}
%\begin{split}
e_R = {\rm Tr}\left[ \left( {\bf Z}^{-1}_{r,H} \otimes {\bf Z}^{-1}_{t,H} + {\bf K}^{-1}_{r,R} \otimes  {\bf S}^* {\bf K}^{-1}_{q,R} {\bf S}^T \right)^{-1} \right]
= \sum^{M}_{n=1} \sigma_{r,H,n} {\rm Tr}\left[ \left({\bf Z}^{-1}_{t,H} + \beta_{R,n} {\bf S}^* {\bf K}^{-1}_{q,R} {\bf S}^T \right )^{-1} \right],
%\end{split}
\end{equation}
where $ \beta_{R,n}\triangleq \frac{\sigma_{r,H,n}}{\delta_{r,R,n}}$. To obtain \eqref{App-I-0}, we have used the rules $ ({\bf A}\otimes {\bf B})^{-1}=  {\bf A}^{-1} \otimes {\bf B}^{-1}$ and \eqref{IV-4}.
%\textcolor{red}{(! I do not know why we should include (28) here. Actually, our paper has a longer length. I think we do not have space to include a same equation here.)}
Since ${\bf S}^* {\bf K}^{-1}_{q,R} {\bf S}^T \in \mathbb{C}^{(N_1+N_2) \times (N_1+N_2)}$, if the designed training sequence ${\bf S}$ makes ${\bf S}^* {\bf K}^{-1}_{q,R} {\bf S}^T$ full rank, the MSE can be arbitrary small by increasing the source power. In this case, the minimum length of ${\bf S}_i$ should satisfy $L_s \geq N_1+N_2$.
%\textcolor{red}{(! This conclusion can be easily checked. If ${\bf S}^* {\bf K}^{-1}_{q,R} {\bf S}^T$ full rank, increasing source power can make the term $\beta_{R,n} {\bf S}^* {\bf K}^{-1}_{q,R} {\bf S}^T$ large enough, then we have a small enough MSE.)}
Otherwise, based on the fact that ${\rm Rank}({\bf A}{\bf B}) \leq \min\{{\rm Rank}({\bf A}),  {\rm Rank}({\bf B})\}$, we must have ${\rm Rank}( {\bf S}^* {\bf K}^{-1}_{q,R} {\bf S}^T) \leq N_1+N_2$. Let us consider the best case scenario, where ${\rm Rank}( {\bf S}^* {\bf K}^{-1}_{q,R} {\bf S}^T ) = L_S$,
%\textcolor{red}{(! Here we use the fact the rank of ${\bf A}{\bf B}$, $r_{AB}$, satisfy $r_{AB} \leq \min(r_A,r_B)$ with $r_A$ and $r_B$ being the rank of ${\bf A}$ and ${\bf B}$, respectively.)}.
In this case, the MSE in \eqref{App-I-0} can be lower bounded by
\begin{equation}\label{App-I-1}
\begin{split}
e_R  \geq &\sum^{M}_{n=1} \sigma_{r,H,n} \left(  \sum^{L_s}_{m=1} \frac{1}{\sigma^{-1}_{t,H,m} + \beta_{R,m} \lambda_{SK,m}} +
 \sum^{N_1+N_2}_{m=L_S+1} \sigma_{t,H,m} \right),
\end{split}
\end{equation}
where $\lambda_{SK,m}$ is the $m$-th element of ${\bm \lambda}({\bf S}^* {\bf K}^{-1}_{q,R} {\bf S}^T)$. Moreover, the eigenvlaues in $\{\sigma_{t,H,m}\}$ and $\{\lambda_{SK,m}\}$ are assumed to be arranged in decreasing order, respectively. To obtain \eqref{App-I-1}, we have use the fact that the function ${\rm Tr}[({\bf A})]$ is a schur convex function with respect to the eigenvalue of $\bf A$ and the following result from \cite{Book_Majorization}
\begin{equation}\nonumber
\begin{split}
{\bm \lambda}({\bf Z}_{r,H}) + {\bm \lambda}({\bf S}^* {\bf K}^{-1}_{q,R} {\bf S}^T) \preccurlyeq {\bm \lambda}({\bf Z}_{r,H} +{\bf S}^* {\bf K}^{-1}_{q,R} {\bf S}^T).
\end{split}
\end{equation}
When the source power is large enough,
{the term $\sum^{L_s}_{m=1} \frac{1}{\sigma^{-1}_{t,H,m} + \beta_{R,m} \lambda_{SK,m}}$ in \eqref{App-I-1} approaches zero.}
Thus, $e_R$ in \eqref{App-I-1} is lower bounded by
%\begin{equation}\label{App-I-3}\nonumber
%\begin{split}
$e_R  \geq \sum^{M}_{n=1} \sigma_{r,H,n}  \sum^{N_1+N_2}_{m=L_S+1} \sigma_{t,H,m}$.
%\end{split}
%\end{equation}

If ${\bf K}_{q,R} = q {\bf I}$, the total MSE in \eqref{App-I-0} can be written as
\begin{equation}\label{App-I-00}
\begin{split}
e_R & = \sum^{M}_{n=1} \sigma_{r,H,n} {\rm Tr}\left[ \left({\bf Z}^{-1}_{t,H} + \tilde{\beta}_{R,n} {\bf S}^* {\bf S}^T \right )^{-1} \right],
\end{split}
\end{equation}
where $\tilde{\beta}_{R,n} \triangleq {\beta}_{R,n}/q$. With a finite power at the source, it is assumed that the optimal solution of ${\bf S}_1$ and ${\bf S}_2$ in \eqref{IV-1} results in the optimal ${\bf S}$ to have a rank of $r\leq N_1+N_2$. By using the SVD decomposition, we assume that the optimal ${\bf S}$ can be decomposed to
\begin{equation}\label{App-I-4}\nonumber
%\begin{split}
{\bf S} = {\bf U}_S {\bm \Sigma}_S {\bf V}^H_S,
%\end{split}
\end{equation}
where ${\bf U}_S$ and ${\bf V}_S$ are matrices of size $(N_1+N_2)\times r$ and $L_S \times r$, respectively. Moreover, ${\bf U}^H_S {\bf U}_S = {\bf I}_r$, ${\bf V}^H_S {\bf V}_S = {\bf I}_r$, and ${\bm \Sigma}_S$ is an $r \times r$ diagonal eigenvalue matrix. The optimal ${\bf S}_1$ and ${\bf S}_2$ can be denoted as
\begin{equation}\label{App-I-4}\nonumber
%\begin{split}
{\bf S}_1 = {\bf U}_{S,1} {\bm \Sigma}_S {\bf V}^H_S~{\rm and}~{\bf S}_1 = {\bf U}_{S,2} {\bm \Sigma}_S {\bf V}^H_S,
%\end{split}
\end{equation}
where ${\bf U}_{S,1} \triangleq {\bf U}_S(1:N_1,:)$ and ${\bf U}_{S,2} \triangleq {\bf U}_S(N_1+1:N_1+N_2,:)$. Subsequently, a new $\tilde{{\bf S}}$ given by $\tilde{{\bf S}} = {\bf U}_S {\bm \Sigma}_S$, can be obtained that achieves the same total MSE as that of the optimal ${\bf S}$, however, with a shorter training sequence length of $L_S = r$. Furthermore, the new optimal $\tilde{{\bf S}}_1 = {\bf U}_{S,1} {\bm \Sigma}_S  $ and $\tilde{{\bf S}}_2 = {\bf U}_{S,2} {\bm \Sigma}_S $ require the same power at the sources nodes compared to the optimal training sequences. This completes the proof of \emph{Lemma~\ref{MAC_length}}.

\section{Proof of Lemma~\ref{MAC_sec_A}}
\label{prof_lemma1}
{
By taking the gradient of $g(\lambda_i)$, we can easily verify that
$g(\lambda_i)$ decreases with $\lambda_i$}. Next, we mainly focus on deriving the upper bound of $\lambda_i$. The source power constraint can be rewritten as
\begin{equation}\label{App-II-1}
\begin{split}
{\rm Tr}({\bf s}_i {\bf s}^H_i)  = {\rm Tr}\left[({\bf X}^i_s + \lambda_i {\bf I})^{-2}{\bf x}_{s,3,i}{\bf x}^H_{s,3,i} \right]
 \leq \frac{\sigma_{s,3,i}}{(\sigma_{s,\text{min},i}+\lambda_i)^2},
\end{split}
\end{equation}
where $\sigma_{s,\text{min},i}$ and $\sigma_{s,3,i}$ are defined in \emph{Lemma~\ref{MAC_sec_A}}.
%${\bf X}^i_s = {\bf X}_{s,1} \otimes {\bf X}_{s,2,i}$, $\sigma_{s,\text{min},i}$ is the smallest eigenvalue of ${\bf X}^i_s$ and $\sigma_{s,3,i}$ is the only eigenvalue of ${\bf x}_{s,3,i}{\bf x}^H_{s,3,i}$, i.e., $\sigma_{s,3,i}= ||{\bf x}_{s,3,i}||^2_2$.
In \eqref{App-II-1}, the inequality is obtained based on the identity
%\begin{equation}\label{App-V-4}
%\begin{split}
${\rm Tr}({\bf A}{\bf B}) \leq \sum_i {\sigma}_{A,i} {\sigma}_{B,i}$~\cite{Lasserre1995}.
%\end{split}
%\end{equation}
Here, ${\sigma}_{A,i}$ and ${\sigma}_{B,i}$ are the eigenvalues of the $n\times n$ matrices ${\bf A}$ and ${\bf B}$, respectively,
{$\{{\sigma}_{A,1},{\sigma}_{A,2},\cdots,{\sigma}_{A,n}\}$ and $\{{\sigma}_{B,1},{\sigma}_{B,2},\cdots,{\sigma}_{B,n}\}$ are arranged in the same order,
}
and the equality
%in \eqref{App-V-4}
is achieved when ${\bf A}$ and ${\bf B}$ are diagonal matrices.
Hence we have
%\begin{equation}\label{App-II-2}\nonumber
%\begin{split}
$\frac{\sigma_{s,3,i}}{(\sigma_{s,\text{min},i}+\lambda_i)^2} \geq \tau_i$,
%\end{split}
%\end{equation}
which further implies
%\begin{equation}\label{App-II-3} \nonumber
%\begin{split}
 $\lambda_i \leq \sqrt{\frac{\sigma_{s,3,i}}{\tau_i}}-\sigma_{s,\text{min},i}$.
%\end{split}
%\end{equation}

\section{Proof of Theorem~\ref{MAC_convex}}
\label{prof_theorem3}
{In \eqref{IV-28}, the feasible set established by the power constraints is convex since the function ${\rm Tr}({\bf E}_i{\bf Q}_S)$ is linear \cite{Boyd2004}}. To prove the convexity of \eqref{IV-28}, it is sufficient to show that the objective function is convex. Without loss of generality, we denote that $f({\bf Q}_S)={\rm Tr}\big[ \big({\bf Z}^{-1}_{t,H}  + \alpha_n  {\bf Q}_S  \big)^{-1} \big]$. According to \cite{Boyd2004}, we can prove the convexity of $f({\bf Q}_S)$ by considering an arbitrary linear combination, given by ${\bf Q}_S = {\bf Q}_{S,1}+t{\bf Q}_{S,2}$, where ${\bf Q}_{S,1} \in {\mathbb{S}}^N_{+}$, ${\bf Q}_{S,2} \in {\mathbb{S}}^N$ and ${\bf Q}_{S,1}+t{\bf Q}_{S,2} \in \mathbb{S}^N_{+}$. By defining $g(t)=f({\bf Q}_{S,1}+t{\bf Q}_{S,2})$, we have
%\begin{equation}\label{App-III-1} \nonumber
%\begin{split}
$g(t) = {\rm Tr}\big[ \big({\bf Z}^{-1}_{t,H}  + \alpha_n  ({\bf Q}_{S,1}+t{\bf Q}_{S,2}) \big )^{-1} \big]$.
%\end{split}
%\end{equation}
Then we obtain
%\begin{equation}\label{App-III-2}
%\begin{split}
$\frac{d g(t)}{d t} = -{\rm Tr}\big(\alpha_n \big({\bf Z}^{-1}_{t,H}  + \alpha_n  ({\bf Q}_{S,1}+t{\bf Q}_{S,2}) \big )^{-2} {\bf Q}_{S,2}  \big)$.
%\end{split}
%\end{equation}
%where ${\bf O} = {\bf Z}^{-1}_{t,H}  + \alpha_n  ({\bf Q}_{S,1}+t{\bf Q}_{S,2}) $.
Based on that,
%by letting ${\bf G}_2 = \alpha_i{\bf C}^H_{t,H} {\bf Q}_{S,2} {\bf C}_{t,H}$,
we can further reach
\begin{equation}\label{App-III-3}
\begin{split}
\frac{d^2 g(t)}{d t^2}
  &=   2 \alpha^2_n{\rm Tr}\left(  \left({\bf Z}^{-1}_{t,H}  + \alpha_n  ({\bf Q}_{S,1}+t{\bf Q}_{S,2}) \right )^{-2} {\bf Q}_{S,2}
  \left({\bf Z}^{-1}_{t,H}  + \alpha_n  ({\bf Q}_{S,1}+t{\bf Q}_{S,2}) \right )^{-1} {\bf Q}_{S,2} \right)\\
  &\geq 0,
\end{split}
\end{equation}
To obtain \eqref{App-III-3}, we use the fact that $\big({\bf Z}^{-1}_{t,H}  + \alpha_n  ({\bf Q}_{S,1}+t{\bf Q}_{S,2}) \big)^{-2}$
and ${\bf Q}_{S,2}  \big({\bf Z}^{-1}_{t,H}  + \alpha_n  ({\bf Q}_{S,1}+t{\bf Q}_{S,2}) \big )^{-1}$\\
$\times {\bf Q}_{S,2} $ are positive semidefinite matrices.
Hence, we conclude that the function $f({\bf Q}_S)$ is convex with respect to the positive semidefinite matrix ${\bf Q}_S$, which further implies that the objective function in \eqref{IV-28} is convex since the sum of multiple convex functions is a still a convex function.

\section{Proof of Theorem~\ref{MAC_sec_B}}
\label{prof_theorem4}
For notation convenience, we define ${\bf D}_0 \triangleq {\bf Z}^{-1}_{t,H_1} + \alpha_n  {\bf S}^*  {\bf S}^T $ and let
\begin{equation}\label{App-IV-1} \nonumber
\begin{split}
{\bf D}_0 =& {\bf Z}^{-1}_{t,H} + \alpha_n  {\bf S}^*{\bf S}^T
 = \left[ \begin{array}{cc}
                           {\bf Z}^{-1}_{t,H_1} + \alpha_n {\bf S}^*_1  {\bf S}^T_1  &  \alpha_n  {\bf S}^*_1  {\bf S}^T_2  \\
                           \alpha_n  {\bf S}^*_2  {\bf S}^T_1  & {\bf Z}^{-1}_{t,H_2} + \alpha_n {\bf S}^*_2 {\bf S}^T_2 \\
                         \end{array}
                       \right]
 \triangleq  \left[
                         \begin{array}{cc}
                           {\bf D}_1 & {\bf D}^H_2 \\
                           {\bf D}_2 & {\bf D}_3 \\
                         \end{array}
                       \right].
\end{split}
\end{equation}
According to the matrix inverse identity, we have
\begin{equation}\label{App-IV-2} \nonumber
\begin{split}
{\bf D}^{-1}_0  = \left[
                         \begin{array}{cc}
                          {\bf A}_0  &  {\bf B}_0  \\
                          {\bf C}_0  & {\bf D}_0   \\
                         \end{array}
                       \right],
\end{split}
\end{equation}
where ${\bf A}_0 = ({\bf D}_1-{\bf D}^H_2 {\bf D}^{-1}_3 {\bf D}_2)^{-1}$, ${\bf B}_0 = -{\bf D}^{-1}_1 {\bf D}^H_2({\bf D}_3-{\bf D}_2 {\bf D}^{-1}_1  {\bf D}^H_2)^{-1}$, ${\bf C}_0 = -{\bf D}^{-1}_3 {\bf D}_2({\bf D}_1-{\bf D}^H_2 {\bf D}^{-1}_3 {\bf D}_2)^{-1}$, and ${\bf D}_0 = ({\bf D}_3-{\bf D}_2 {\bf D}^{-1}_1  {\bf D}^H_2 )^{-1}$.
Subsequently, we have that
%\begin{equation}\label{App-IV-3} \nonumber
%\begin{split}
${\rm Tr}({\bf D}^{-1}_0) =  {\rm Tr}\left[({\bf D}_1-{\bf D}^H_2 {\bf D}^{-1}_3 {\bf D}_2)^{-1}\right] +$ \\ $
 {\rm Tr}\left[({\bf D}_3-{\bf D}_2 {\bf D}^{-1}_1  {\bf D}^H_2 )^{-1}\right]$.
%\end{split}
%\end{equation}
Since ${\bf D}^H_2 {\bf D}^{-1}_3 {\bf D}_2\succeq {\bf 0}$ and ${\bf D}_2 {\bf D}^{-1}_1  {\bf D}^H_2 \succeq {\bf 0}$, the following inequalities hold
%\begin{equation}\label{App-IV-4} \nonumber
%\begin{split}
${\bf D}_1-{\bf D}^H_2 {\bf D}^{-1}_3 {\bf D}_2  \preceq {\bf D}_1~{\rm and}~
{\bf D}_3-{\bf D}_2 {\bf D}^{-1}_1  {\bf D}^H_2   \preceq {\bf D}_3$.
%\end{split}
%\end{equation}
This result further implies that
\begin{equation}\label{App-IV-5}
\begin{split}
\left({\bf D}_1-{\bf D}^H_2 {\bf D}^{-1}_3 {\bf D}_2 \right)^{-1}  \succeq {\bf D}^{-1}_1~{\rm and}~
\left({\bf D}_3-{\bf D}_2 {\bf D}^{-1}_1  {\bf D}^H_2 \right)^{-1}   \succeq {\bf D}^{-1}_3,
\end{split}
\end{equation}
where in obtaining the above we have used the fact that the matrices ${\bf D}_1$, ${\bf D}_3$, ${\bf D}_1-{\bf D}^H_2 {\bf D}^{-1}_3 {\bf D}_2$ and ${\bf D}_3-{\bf D}_2 {\bf D}^{-1}_1  {\bf D}^H_2$ are positive semidefinite. From \eqref{App-IV-5}, we have
\begin{equation}\label{App-IV-6}
\begin{split}
{\rm Tr}\left[\left({\bf D}_1-{\bf D}^H_2 {\bf D}^{-1}_3 {\bf D}_2 \right)^{-1}\right]  \succeq {\rm Tr}\left({\bf D}^{-1}_1\right)~{\rm and}~
{\rm Tr}\left[\left({\bf D}_3-{\bf D}_2 {\bf D}^{-1}_1  {\bf D}^H_2 \right)^{-1}\right]   \succeq {\rm Tr}\left({\bf D}^{-1}_3\right).
\end{split}
\end{equation}
Hence, if ${\bf D}_2 = {\bf 0}$, i.e., ${\bf S}^*_2  {\bf S}^T_1= {\bf 0}$, the value of the objective function in \eqref{IV-27} can be always reduced.

Next, it is shown that for any ${\bf S}_1$ and ${\bf S}_2$, letting ${\bf S}^*_2  {\bf S}^T_1= {\bf 0}$ does not increase the need for power at the source nodes. Since in \eqref{IV-27}, the values of the objective function and the power constraints are only affected by ${\bf S}^*_1  {\bf S}^T_1$ and ${\bf S}^*_2  {\bf S}^T_2$, the optimal ${\bf S}_1$ and ${\bf S}_2$ can be determined as
\begin{equation}\label{App-IV-7}
\begin{split}
\bar{{\bf S}}_1 = {\bf U}_{t, s_i} {\bm \Sigma}_{s_i} {\bf V}^H_{s_i},
\end{split}
\end{equation}
where ${\bf V}_{s_i}$ can be any matrix satisfying ${\bf V}^H_{s_i} {\bf V}_{s_i} = {\bf I} $. It is worth noting that since $L_S \geq N_1+N_2$ based on \emph{Lemma~\ref{MAC_length}}, one can always find a specific ${\bf V}_{s_1}$ and ${\bf V}_{s_2}$ such that
\begin{equation}\label{App-IV-77}
\begin{split}
{\bf V}^H_{s_1} {\bf V}_{s_2} = {\bf 0},
\end{split}
\end{equation}
which further results in ${\bf S}^*_2  {\bf S}^T_1= {\bf 0}$.
In this case, we do not change the value of ${\bf S}^*_i  {\bf S}^T_i$ and the power constraint, while decrease the value of the objective function according to \eqref{App-IV-6}.

It is noticed that the orthogonal training sequences has also been proven to be optimal for the cascaded channel estimation in two-way relaying system in \cite{Kim_cl2013, Rong_tsp2013}. Here, we show similar results hold for individual channel estimation. With the optimal condition ${\bf S}^*_1  {\bf S}^T_2=0$, the optimization problem in \eqref{IV-27} can be decomposed into two subproblems given by
\begin{equation}\label{App-IV-8}
\begin{split}
 \min_{{\bf S}_i}~~ & \sum^{M}_{n=1} \sigma_{r,H,n} {\rm Tr}\left[ \left({\bf Z}^{-1}_{t,H_i} + \alpha_n  {\bf S}^*_i  {\bf S}^T_i  \right )^{-1} \right]  \\
{\rm s.t.} ~~~ & {\rm Tr}({\bf S}^*_i {\bf S}^T_i) \leq \tau_i,~i=1,2
\end{split}
\end{equation}
Based on the results derived for point-to-point MIMO systems \cite{Bjornson2010}, we obtain that the unitary matrix ${\bf U}_{t, s_i}$ given in  \eqref{App-IV-7} should be of the form ${\bf U}_{t, s_i} = {\bf U}^*_{t, H_i}$, and ${\bf V}_{s_1}$ should satisfy \eqref{App-IV-77}. Then, solving \eqref{App-IV-8} just reduces to solving the following power allocation problem
\begin{equation}\label{App-IV-9}\nonumber
\begin{split}
 \min_{\sigma_{s_i,m}}~~ &\sum^{M}_{n=1} \sigma_{r,H,n} \sum^{N_i}_{m=1} \frac{\sigma_{t,H_i,m}}{1+\alpha_n \sigma_{t,H_i,m} \sigma^2_{s_i,m}}   \\
{\rm s.t.} ~~~& \sum^{N_i}_{m=1}\sigma^2_{s_i,m}  \leq \tau_i
\end{split}
\end{equation}
where the optimal $\sigma_{s_i,m}$ can be obtained via the water-filling approach in \eqref{App-IV-10}.% \cite{Bjornson2010}.

\section{Proof of Lemma~\ref{BC_sec_A}}
\label{prof_lemma3}
Note that when eigenvalues of ${\bf Z}_{t,G}$, ${\bf Z}_{r,G_i}$, ${\bf K}_{q,i}$ and ${\bf K}_{r,i}$ are small, $\sigma_{r,G_i,n}$ becomes very small compared to ${\bf Z}^{-1}_{t,G}$, $\frac{1}{\delta_{r,1,n}q_1}{\bf S}^*_R {\bf S}^T_R$ and $\frac{1}{\delta_{r,2,n}}{\bf S}^*_R {\bf K}^{-1}_{q,2}{\bf S}^T_R$. Based on the inverse approximation rule in \cite[Eq. (167)]{Cookbook}, the original problem can be approximated as
\begin{equation}\label{App-V-1}
\begin{split}
\min_{{\bf S}_R} & ~~ \sum^{N_1}_{n=1} \sigma_{r,G_1,n} {\rm Tr}\left[ {\bf Z}_{t,G} - \tilde{\beta}_{1,n}  {\bf Z}_{t,G} {\bf S}^*_R {\bf S}^T_R {\bf Z}_{t,G} \right] +
 \sum^{N_2}_{n=1}\sigma_{r,G_2,n} {\rm Tr}\left[ {\bf Z}_{t,G} - \beta_{2,n}  {\bf Z}_{t,G} {\bf S}^*_R {\bf K}^{-1}_{q,2} {\bf S}^T_R {\bf Z}_{t,G} \right] \\
{\rm s.t.}& ~~~ {\rm Tr}({\bf S}^*_R {\bf S}^T_R) \leq \tau_R.
\end{split}
\end{equation}
By defining a new variable ${\bf P}\triangleq{\bf Z}_{t,G} {\bf S}^*_R$, \eqref{App-V-1} is equivalent to
\begin{equation}\label{App-V-2}\nonumber
\begin{split}
\max_{{\bf S}_R} & ~~ \sum^{N_1}_{n=1}  \sigma_{r,G_1,n} {\rm Tr}\left[  \tilde{\beta}_{1,n}  {\bf P} {\bf P}^H \right] +
 \sum^{N_2}_{n=1} \sigma_{r,G_2,n} {\rm Tr}\left[ \beta_{2,n}  {\bf P} {\bf K}^{-1}_{q,2} {\bf P}^H  \right] \\
{\rm s.t.}& ~~~ {\rm Tr}\left[{\bf P} {\bf P}^H ({\bf Z}^{H}_{t,G})^{-1} {\bf Z}^{-1}_{t,G}  \right] \leq \tau_R.
\end{split}
\end{equation}
Assuming that the SVD decomposition of ${\bf P}$ is given by ${\bf U}_P {\bm \Sigma}_P {\bf V}^H_P$ with the eigenvalues in ${\bm \Sigma}_P$ arranged in decreasing order, we have
\begin{equation}\label{App-V-3}\nonumber
\begin{split}
\max_{{\bf S}_R} & ~~ \sum^{N_1}_{n=1} \sigma_{r,G_1,n}{\rm Tr}\left[  \tilde{\beta}_{1,n}  {\bm \Sigma}^2_P  \right] +
 \sum^{N_2}_{n=1}\sigma_{r,G_2,n} {\rm Tr}\left[ \beta_{2,n}  {\bm \Sigma}^{2}_P {\bf V}^H_P  {\bf V}_{q,2}{\bm \Delta}^{-1}_{q,2} {\bf V}^H_{q,2} {\bf V}_P  \right] \\
{\rm s.t.}& ~~~ {\rm Tr}\left[{\bm \Sigma}^2_P {\bf U}^H_P {\bf U}_{t,G}{\bm \Sigma}^{-2}_{t,G} {\bf U}^H_{t,G} {\bf U}_P  \right] \leq \tau_R.
\end{split}
\end{equation}
We observe that the unitary matrices ${\bf U}_P$ and ${\bf V}_P$ only affect the power constraint and the objective function, respectively. Based on the trace inequality identity in \cite[Eq. (3)]{Lasserre1995}, we have
\begin{equation}\label{App-V-44}\nonumber
\begin{split}
 {\rm Tr}\big[ \beta_{2,n}  {\bm \Sigma}^{2}_P {\bf V}^H_P  {\bf V}_{q,2}{\bm \Delta}^{-1}_{q,2} {\bf V}^H_{q,2} {\bf V}_P  \big] &\leq {\rm Tr}\big[ \beta_{2,n}  {\bm \Sigma}^{2}_P {\bm \Delta}^{-1}_{q,2}  \big]\notag \\
 {\rm Tr}\big[{\bm \Sigma}^2_P {\bf U}^H_P {\bf V}_{t,G}{\bm \Sigma}^{-2}_{t,G} {\bf V}^H_{t,G} {\bf U}_P  \big] &\geq {\rm Tr}\big[{\bm \Sigma}^2_P {\bm \Sigma}^{-2}_{t,G} \big],
\end{split}
\end{equation}
where the eigenvalues in ${\bm \Delta}_{q,2}$ and ${\bm \Sigma}_{t,G}$ are arranged in increasing order and decreasing order, respectively,
and the equalities are achieved when ${\bf V}_P = {\bf V}_{q,2}$ and ${\bf U}_P ={\bf U}_{t,G}$. Hence, we obtain the optimal ${\bf S}^*_R$ as
%\begin{equation}\label{App-V-5}\nonumber
%\begin{split}
${\bf S}^*_R  = {\bf Z}^{-H}_{t,G} {\bf U}_{t,G} {\bm \Sigma}_P {\bf V}^H_{q,2}
 = {\bf U}_{t,G} {\bm \Sigma}_{t,G} {\bm \Sigma}_P {\bf V}^H_{q,2}$,
%\end{split}
%\end{equation}
which further leads to \eqref{IV-43}.

With the structure of training sequence given in \eqref{IV-43}, the original optimization problem in \eqref{IV-42} is reduced to the following power allocation problem
\begin{equation}\label{App-V-6}
\begin{split}
\min_{\sigma_{s,R,n},\forall n} & \sum^{N_1}_{n=1}  \sum^{M}_{m=1}  \frac{\sigma_{t,G,m}\sigma_{r,G_1,n}}{1+\tilde{\beta}_{1,n} \sigma_{t,G,m} \sigma_{s,R,m} }+
  \sum^{N_2}_{n=1}  \sum^{M}_{m=1} \frac{\sigma_{r,G_2,n}\sigma_{t,G,m}}{1+{\beta}_{2,n} \sigma_{t,G,m} \sigma_{s,R,m} \delta^{-1}_{q,2,m} } \\
{\rm s.t.}~~& \sum^{M}_{m=1} \sigma_{s,R,m} \leq \tau_R
\end{split}
\end{equation}
The lagrangian function of \eqref{App-V-6} can be written as
\begin{equation}\label{App-V-7}\nonumber
\begin{split}
\mathcal{L} =  \sum^{N_1}_{n=1}  \sum^{M}_{m=1}  \frac{\sigma_{t,G,m}\sigma_{r,G_1,n}}{1+\tilde{\beta}_{1,n} \sigma_{t,G,m} \sigma_{s,R,m} }
+\sum^{N_2}_{n=1}  \sum^{M}_{m=1} \frac{\sigma_{r,G_2,n}\sigma_{t,G,m}}{1+{\beta}_{2,n} \sigma_{t,G,m} \sigma_{s,R,m} \delta^{-1}_{q,2,m} } + \lambda (\sum^{M}_{m=1} \sigma_{s,R,m} - \tau_R),
\end{split}
\end{equation}
where $\lambda$ is lagrangian multiplier. Based on the KKT condition, we obtain \eqref{IV-44}. Then, by setting $\sigma_{s,R,m}=0$, we obtain the range of $\lambda$ as shown in \emph{Lemma~\ref{BC_KKT}}.
%\textcolor{red}{(! Here just a mathematic trick to obtain an upper bound. )}

%============================References=========================
%\bibliographystyle{IEEEtran}
%\bibliography{IEEEabrv,two_way_relay}

%============================figures=========================

\end{document}